\newif\iftwo %
\twotrue 

\iftwo

\providecommand{\two}[2]{#1}
\documentclass[twoside]{IEEEtran}

\else

\providecommand{\two}[2]{#2}
\documentclass[12pt,draftcls,onecolumn]{IEEEtran}

\fi

\usepackage{vkmacros}

\newif\ifmapx
\mapxfalse

\newif\ifshowtodo
\showtodotrue

\NewEnviron{apxonly}
{
 \ifmapx\expandafter
 \addtocounter{equation}{-1}
\begin{subequations}\renewcommand\theequation{\theparentequation\alph{equation}}
  \begingroup\color{violet!80!black}
 \BODY
 \endgroup
\end{subequations}
 \fi
}

\title{Successive Refinement of Abstract Sources }

\author{
Victoria Kostina, Ertem Tuncel 
\thanks{
V. Kostina (e-mail: \href{mailto:vkostina@caltech.edu}{vkostina@caltech.edu}) is with California Institute of Technology.
E. Tuncel (e-mail: \href{mailto:ertem.tuncel@ucr.edu}{ertem.tuncel@ucr.edu}) is with University of California, Riverside.
This work was supported in part by the National Science Foundation (NSF) under Grant CCF-1566567. It was presented in part at ISIT~2017 \cite{kostina2017succrefISIT}.}
}

\IEEEoverridecommandlockouts
\begin{document}
\maketitle
\begin{abstract}
In successive refinement of information, the decoder refines its representation of the source progressively as it receives more encoded bits.  
The rate-distortion region of successive refinement describes the minimum rates required to attain the target distortions at each decoding stage. In this paper, we derive a parametric characterization of the rate-distortion region for successive refinement of abstract sources. Our characterization extends Csisz\'ar's result \cite{csiszar1974extremum} to successive refinement, and generalizes a result by Tuncel and Rose \cite{tuncel2003computation}, applicable for finite alphabet sources, to abstract sources. This characterization spawns a family of outer bounds to the rate-distortion region. It also enables an iterative algorithm for computing the rate-distortion region, which generalizes Blahut's algorithm to successive refinement. Finally, it leads a new nonasymptotic converse bound. In all the scenarios where the dispersion is known, this bound is second-order optimal. 

 In our proof technique, we avoid Karush-Kuhn-Tucker conditions of optimality, and we use basic tools of probability theory. We leverage the Donsker-Varadhan lemma for the minimization of relative entropy on abstract probability spaces.  

\end{abstract}


\section{Introduction}
For a source random variable $X \in \mathcal X$ and a distortion measure $\sd \colon \mathcal X \times \mathcal Y \mapsto [0, + \infty)$, where $\mathcal X$ and $\mathcal Y$ are abstract sets (source and reproduction alphabets), the classical informational rate-distortion function is defined as the following minimal mutual information quantity:
\begin{equation}
 R(d)  \triangleq \inf_{ 
\substack{ P_{Y | X} \colon \mathcal X \mapsto \mathcal Y\\  
\E{\sd(X, Y)} \leq d
}} I(X; Y) \label{eq:RRd}
\end{equation}
This convex optimization problem rarely has an explicit solution. The following result provides a parametric representation:

\begin{thm} [Parametric representation of $R(d)$  \cite{csiszar1974extremum}]
 Assume that the following conditions are met. 
 
 \begin{enumerate}[(A)]
  \item  \label{item:a1}
$d_{\min} < d < d_{\max}$, where 
\begin{align}
 d_{\min} &\triangleq \inf \left\{ d\colon ~ R(d) < \infty \right\} \label{sc:dmin}\\
 d_{\max} &\triangleq \inf\left\{ d \colon ~ R(d) \text{ is constant on } (d_{\max}, \infty) \right\}
\end{align}

 \item \label{item:b1} There exists a transition probability kernel $P_{Y^\star | X}$ that attains the infimum in \eqref{eq:RRd}. 
\end{enumerate}

Then, it holds that
\begin{equation}
R(d)  = \max_{\alpha(x), \lambda}\left\{- \E{ \log  {\alpha(X)} } - \lambda d\right\} \label{eq:RR(d)csiszar}
\end{equation}
where the maximization is over $\alpha(x) \geq 0$ and $\lambda\geq 0$ satisfying the constraint
\begin{equation}
\E{ \frac {\exp\left(  - \lambda \mathsf d(X, y)\right) } {\alpha(X)} } \leq 1 ~ \forall y \in \mathcal Y \label{eq:csiszarg}.
\end{equation}
Furthermore, in order for $P_{Y^\star | X}$ to achieve the infimum in \eqref{eq:RRd}, it is necessary and sufficient that
\begin{equation}
\frac{d P_{X | Y^\star = y}}{dP_{X}}(x) =  \frac {\exp (-\lambda^\star \sd(x, y)) } {\alpha(x) }, \label{eq:PYcondstar}
\end{equation}
where\footnote{The differentiability of $R(d)$ is assured by the assumptions that the distortion measure $\sd$ cannot take the value $+\infty$ and that there exists a $P_{Y^\star | X}$ attaining the infimum in \eqref{eq:RRd}  \cite[p. 69]{csiszar1974extremum}. If we allow $\sd$ to take the value $+\infty$, then it is possible that $R(d)$ is not differentiable at some $d$. In that case, \thmref{thm:csiszarg} will hold verbatim replacing $\lambda^\star$ by the negative slope of any tangent to $R(d)$ at $d$. With this easy extension in mind, we choose to limit our attention to finite-valued distortion measures to ensure differentiability. Note also that while $P_{Y^\star | X}$ need not be unique, $\alpha^\star(x)$ is (and therefore, through \eqref{eq:PYcondstar}, so is $P_{X^\star | Y}$); this is a consequence of differentiability of $R(d)$  \cite[p. 69]{csiszar1974extremum}.}
\begin{equation}
\lambda^\star = - R^\prime(d),
\end{equation}
and $0 \leq \alpha(x) \leq 1$ satisfies \eqref{eq:csiszarg}. Finally, the choice
\begin{align}
\alpha^\star(x) &=  \E{\exp ( -\lambda^\star \sd(x,  Y^\star) ) }, 
\end{align}
satisfies both \eqref{eq:csiszarg} and \eqref{eq:PYcondstar}; thus $(\alpha^\star(x), \lambda^\star)$ is the maximizer of \eqref{eq:RR(d)csiszar}.
\label{thm:csiszarg}
\end{thm}
\vspace*{-.3cm}
In \eqref{eq:PYcondstar}, $\frac {dP}{dQ}$ denotes the Radon-Nykodym derivative; if $P$ and $Q$ are both discrete / continuous probability distributions, $\frac {dP}{dQ}$ is simply the ratio of corresponding probability mass / density functions. \thmref{thm:csiszarg} applies to the much more general setting of abstract probability spaces.   It was Csisz\'ar \cite{csiszar1974extremum} who formulated and proved \thmref{thm:csiszarg} in this generality.\footnote{Even more generally, Csisz\'ar \cite{csiszar1974extremum} showed that \eqref{eq:RR(d)csiszar} continues to hold even if the infimum in \eqref{eq:RRd} is not attained by any conditional probability distribution.} For finite alphabet sources,  the parametric representation of $R(d)$  is contained in Shannon's paper \cite{shannon1959coding};  Gallager's \cite[Th. 9.4.1]{gallager1968information} and Berger's \cite{berger1971rate} texts include the parametric representation of $R(d)$ for discrete and continuous sources. Csisz\'ar and K\"orner's book \cite[Th. 8.7]{csiszar2011information} presents a derivation of the parametric representation of the discrete rate-distortion function that employs variational principles. 

The parametric representation of $R(d)$ plays a key role in the Blahut algorithm \cite{blahut1972computation} for computing the rate-distortion function. 
For difference distortion measures, $\sd(x, y) = \sd(x-y)$, a certain choice of $(\alpha(x),\lambda)$ in \eqref{eq:RR(d)csiszar} leads to the Shannon lower bound \cite{shannon1959coding}, a particularly simple, explicit lower bound to the rate-distortion function, which offers nice intuitions and which is known to be tight in the limit $d \downarrow 0$. Leveraging \thmref{thm:csiszarg}, a generalization of Shannon's lower bound to abstract probability spaces was recently proposed \cite{kostina2016slb,kostina2016lowd}. Furthermore, given $(P_X, \sd)$, the $\sd$-tilted information, defined for each realization $x \in \mathcal X$ through the solution to \eqref{eq:RR(d)csiszar} as
\begin{equation}
 \jmath_{\sd}(x, d) \triangleq - \log \alpha^\star(x) - \lambda^\star d, \label{eq:dtilted}
\end{equation}
governs the nonasymptotic fundamental limits of lossy compression \cite{kostina2011fixed}, where the subscript $\sd$ emphasizes the distortion measure used.  

In this paper, we state and prove a generalization of \thmref{thm:csiszarg} to successive refinement of abstract alphabet sources. If the source is successively refinable, that is, if optimal successive coding achieves the respective rate-distortion functions at each decoding stage, our result recovers the representation in \thmref{thm:csiszarg}. 
Our characterization refines a prior finite alphabet result by Tuncel and Rose~\cite[Theorem 4]{tuncel2003computation} and extends it to abstract probability spaces. Our general setting necessitates the use of the mathematical tools fundamentally different from  the standard convex optimization tools (Karush-Kuhn-Tucker conditions) that can be used to solve the finite alphabet case, as carried out in  \cite{tuncel2003computation}. We leverage the Donsker-Varadhan characterization of the minimum relative entropy, and, to show the necessary optimality conditions, we compare a tentative solution to a perturbation by a carefully selected auxiliary distribution.

The new characterization of rate-distortion function for successive refinement on abstract alphabets allows us to identify the key random variable describing the nonasymptotic fundamental limits of successive refinement, and to show a new nonasymptotic converse bound.   In all the scenarios where the dispersion of successive refinement is known \cite{no2016strong,zhou2017successive}, this bound is second-order optimal. 

The new characterization also enables an iterative algorithm, which can be used to compute an accurate approximation to the rate-distortion function of successive refinement, even if the source and reproduction alphabets are not discrete. We prove that when initialized appropriately, the algorithm converges to the true value of rate-distortion function with speed $\bigo{\frac 1 k}$, where $k$ is the iteration number. The algorithm  can be viewed as a generalization of Blahut's algorithm \cite{blahut1972computation} and its extension to successive refinement  by Tuncel and Rose \cite{tuncel2003computation} for discrete alphabets. Methods to compute the capacity and rate-distortion functions for continuous alphabets were proposed in \cite{chang1988calculatinginfinite} and \cite{rose1994mapping}.

The rest of the paper is organized as follows. The main result of the paper characterizing the abstract rate-distortion function (\thmref{thm:parametric}) is presented in \secref{sec:main}. The main nonasymptotic converse result, \thmref{thm:c}, is shown in \secref{sec:nonasymptotic}. A proof of \thmref{thm:csiszarg}, which streamlines Csisz\'ar's argument \cite{csiszar1974extremum}, is presented in \secref{sec:csiszar}. The proof of \thmref{thm:parametric}, which leverages the ideas presented in \secref{sec:csiszar} and in \cite{tuncel2003computation}, is presented in \secref{sec:proof}. \secref{sec:algo} discusses the iterative algorithm for computation of rate-distortion function of successive refinement.

Throughout the paper, $\mathbb R_+ = [0, +\infty)$ is the positive real line; $P_X$-a.e. $x$ stands for `almost every $x$', i.e. 'except on a set with total $P_X$ measure 0'; $P_X \to P_{Y | X} \to P_{Y}$ signifies that $P_Y$ is the distribution observed at the output of random transformation $P_{Y|X}$ when the input is distributed according to $P_X$, i.e. $P_Y$ is the marginal of $P_X P_{Y|X}$. When we say that a random variable $X$ takes values in a set $\mathcal X$, we understand that $\mathcal X$ comes together with its $\sigma$-algebra $\mathscr X$, forming a measurable space $(\mathcal X, \mathscr X)$. Throughout the paper, we assume that 
all $\sigma$-algebras contain singletons (this is true for any countably separated $\sigma$-algebra). For two measurable spaces $(\mathcal X, \mathscr X)$ and $(\mathcal Y, \mathscr Y)$, a \emph{transition probability kernel} from $(\mathcal X, \mathscr X)$ into $(\mathcal Y, \mathscr Y)$ is a mapping $\kappa \colon \mathcal X \times \mathscr Y \mapsto [0, 1]$ such that {\it{(i)}} the mapping $x \mapsto \kappa (x, B)$ is $\mathscr X$-measurable for every $B \in \mathscr Y$, and  {\it{(ii)}} the mapping  $B \mapsto \kappa (x, B)$ is a probability measure on $(\mathcal Y, \mathscr Y)$ for every $x \in \mathcal X$.

\section{Characterization of rate-distortion function}
\label{sec:main}
Consider the source random variable $X \in \mathcal X$ and two (possibly different) distortion measures  $\sd_1 \colon \mathcal X \times \mathcal Y_1 \mapsto [0, + \infty)$ and $\sd_2 \colon \mathcal X \times \mathcal Y_2 \mapsto [0, + \infty)$, quantifying the accuracy of lossy compression at the first and the second stages, respectively.
An $(M_1, M_2, d_1, d_2)$ average distortion code for $(P_X, \sd_1, \sd_2)$ is a pair of encoders
\begin{align}
 &\sf_1 \colon \mathcal X \mapsto \{1, \ldots, M_1\} \label{eq:enc1}\\
  &\sf_2 \colon \mathcal X \mapsto \left\{1, \ldots, \left \lfloor {M_2}/{M_1} \right \rfloor \right\} \label{eq:enc2}
\end{align}
and decoders
\begin{align}
  &\sg_1 \colon  \{1, \ldots, M_1\} \mapsto \mathcal Y_1 \label{eq:dec1}\\
  &\sg_2 \colon \{1, \ldots, M_1\}  \times \left\{1, \ldots, \left \lfloor {M_2}/{M_1} \right \rfloor \right\}  \mapsto\mathcal Y_2 \label{eq:dec2}
\end{align}
such that
\begin{align}
\E{\sd_1(X, \sg_1 (\sf_1(X))) } &\leq d_1, \label{eq:d1av}\\
 \E{\sd_2(X, \sg_2 (\sf_1(X), \sf_2(X) )) } &\leq d_2. \label{eq:d2av}
\end{align}
For the successive refinement of $n$ i.i.d. copies of $X$ with separable distortion measures $\sd_1^{(n)}(x^n, y^n) = \frac 1 n \sum_{i = 1}^n \sd_1(x_i, y_i)$,  $\sd_2^{(n)}(x^n, y^n) = \frac 1 n \sum_{i = 1}^n \sd_2(x_i, y_i)$, we say that the distortions $(d_1, d_2)$ are asymptotically attainable with rates  $(R_1, R_2)$ at first and second stages if there exists a sequence of $(M_1, M_2, d_1, d_2)$ average distortion codes for $(P_{X^n}, \sd_1^{(n)}, \sd_2^{(n)})$ with
\begin{align}
\limsup_{n \to \infty} \frac 1 n \log M_1 &\leq R_1,\\
 \limsup_{n \to \infty} \frac 1 n \log M_2 &\leq R_2.
\end{align}
Rimoldi \cite{rimoldi1994successive} showed that for the discrete memoryless source, the distortions $(d_1, d_2)$ are asymptotically attainable with rates  $(R_1, R_2)$ at first and second stages if and only if
\begin{equation}
\begin{aligned}
 I(X; Y_1) &\leq R_1 								&\E{\sd_1(X, Y_1) }& \leq d_1\\
I(X; Y_1, Y_2) &\leq R_2 							&\E{\sd_2(X, Y_2) }& \leq d_2, 
\label{eq:reg}
\end{aligned}
\end{equation}
where here and in the sequel, $R_2$ refers to the {\em total rate} at both stages (see Effros \cite{effros1999distortion} for a generalization to continuous alphabets and stationary sources). 
It is convenient to consider the following equivalent representation of the boundary of the set in \eqref{eq:reg}:
\begin{align}
\label{eq:RRdsr}
 R_2(d_1, d_2, R_1) \triangleq \inf \left\{R_2 \colon (d_1, d_2, R_1, R_2) \text{ satisfy } \eqref{eq:reg} \right\}.
\end{align}
Henceforth, we refer to the function $R_2(d_1, d_2, R_1) \colon \mathbb R_+^3 \mapsto \mathbb R_+$ as the \emph{ second stage rate-distortion function}. 
It represents the minimum asymptotically achievable total rate compatible with rate $R_1$ at the first stage and at-stage distortions $d_1, d_2$.  
For any achievable $(R_1, R_2, d_1, d_2)$, the following bound in terms of the standard rate-distortion function in \eqref{eq:RRd} clearly holds:
\begin{align}
R_1 &\geq R_{\sd_1}(d_1)  \\
R_2 &\geq R_{\sd_2}(d_2) 
\end{align}
where $R_{\sd_1}(\cdot)$ and $R_{\sd_2}(\cdot)$ denote the rate-distortion functions for distortion measures $\sd_1$ and $\sd_2$, respectively. In \figref{fig:reg}, $(d_1, d_2)$ are fixed, and the region of achievable $(R_1, R_2)$ is greyed out; $R_2(d_1, d_2, R_1)$ is its boundary drawn in red.
If the point $(R_{\sd_1}(d_1), R_{\sd_2}(d_2))$ is attainable, the source is said to be \emph{successively refinable} \cite{equitz1991successive} at $(d_1, d_2)$.

\begin{figure}[htp]
\centering
    \epsfig{file=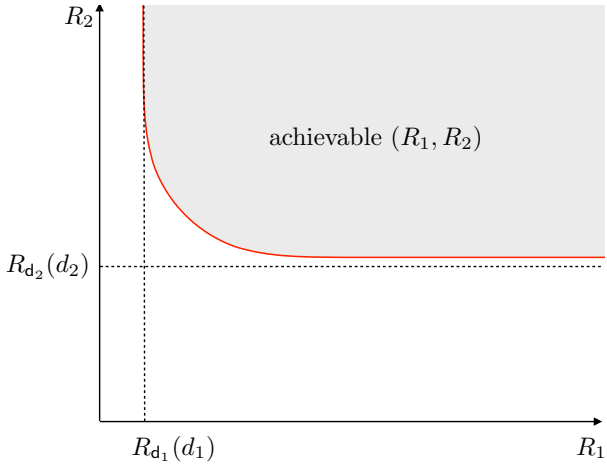,width=\two{.9}{.45}\linewidth}
 \caption[]{The rate-distortion region for successive refinement, for fixed $d_1, d_2$. Note that %
 if $\sd_1 = \sd_2$, and $d_2 < d_1$, then $R_2(d_1, d_2, R_1) = R(d_2)$ is attained at some $R_1 < \infty$. 
 } \label{fig:reg}
\end{figure}

Throughout the paper, we assume that the following conditions are met.  
 \begin{enumerate}[(a)]
  \item  \label{item:a}
 $ R_2(d_1, d_2, R_1)$ is finite in some nonempty region $ \subseteq \mathbb R_+^3$. 
 \item \label{item:b} There exist transition probability kernels $P_{Y^\star_1 | X}$ and $P_{Y_2^\star | XY_1^\star}$ that attain the infimum in \eqref{eq:RRdsr}. 
\end{enumerate}

The mild assumption \eqref{item:b} is always satisfied, for example, if $\mathcal Y_1, \mathcal Y_2$ are finite;  if $\mathcal X$ is Polish,  $\mathcal Y_1, \mathcal Y_2$ are compact metric, and distortion measures $\sd_1$, $\sd_2$ are jointly continuous; and if $\mathcal X = \mathcal Y_1 = \mathcal Y_2$ are Euclidean spaces with $\sd_1(x, y) \to \infty$, $\sd_2(x, y) \to \infty$ as $\|x - y\| \to \infty$ \cite{csiszar1974extremum}.

The second stage rate-distortion function $R_2(d_1, d_2, R_1)$ is nondecreasing and jointly convex in $(d_1, d_2, R_1)$ (see \lemref{lem:convex} in \secref{sec:proof} below). The  region of $(d_1, d_2, R_1)$ where the constraints are satisfied with equality is defined as follows. 
\begin{align}
\label{eq:omega}
\Omega \triangleq \{&
(d_1, d_2, R_1) \in \mathbb R_+^3 \colon  \forall (\epsilon_1, \epsilon_2, \epsilon_3 ) > 0 \colon 
\two{\\ \notag}{}
 &R_2(d_1 + \epsilon_1, d_2+\epsilon_2, R_1+\epsilon_3) < R_2(d_1, d_2, R_1) <  \infty \}. 
\end{align}
In the important special case of $\sd_1=\sd_2$, 
\begin{align}
\!\!\!\! \Omega = \left\{ (d_1, d_2, R_1) \colon R(d_1) < R_1 < R(d_2), ~ d_1 \leq d_{\max} \right\}, 
\end{align}
where $d_{\max}$ is the smallest positive scalar such that $R(d)$ is constant on $(d_{\max}, \infty)$.

Since $ R_2(d_1, d_2, R_1)$ is convex in its input, each point $(d_1, d_2, R_1) \in \mathbb R_+^3$ on the curve can be parametrized via the supporting hyperplane $h - \lambda_1^\star d_1 - \lambda_2^\star d_2 - \nu_1^\star R_1 = 0$. Here $h$ is the is the distance of the hyperplane from the origin, and the triple $(\lambda_1^\star, \lambda_2^\star, \nu_1^\star)$ defines the normal vector to the hyperplane. Thus, to each $(d_1, d_2, R_1) \in \mathbb R_+^3$ there corresponds a triplet $(\lambda_1^\star, \lambda_2^\star, \nu_1^\star) \in \mathbb R_+^3$ such that for some $h \in \mathbb R_+$, the hyperplane $h - \lambda_1^\star d_1 - \lambda_2^\star d_2 - \nu_1^\star R_1 = 0$ is tangent to $R_2(d_1, d_2, R_1)$ at $(d_1, d_2, R_1)$.

Before we state our main result, we present the following notation. For measurable functions $\beta_1 \colon \mathcal X \mapsto \mathbb R_+$, $\beta_2 \colon \mathcal X \times \mathcal Y_1 \mapsto \mathbb R_+$ and nonnegative numbers $\lambda_1, \lambda_2, \nu_1$, denote
\begin{align}
\Sigma_2(y_1, y_2) &\triangleq \E{ \frac {\exp \left(-\frac{\lambda_1}{1+\nu_1} \sd_1(X, y_1) - \lambda_2 \sd_2(X, y_2)  \right)} {\beta_1(X) \beta_2(X| y_1)^{\frac{\nu_1}{1+\nu_1}} } } \\
\Sigma_1(y_1) &\triangleq \E{ \frac{\exp \left(-\frac{\lambda_1}{1+\nu_1} \sd_1(X, y_1)  \right)}{\beta_1(X)\beta_2(X| y_1)^{-\frac{1}{1+\nu_1}} }  }. 
\end{align}
The quantities $\Sigma_1(y_1)$ and $\Sigma_{2}(y_1,y_2)$ generalize the expectation on the left side of \eqref{eq:csiszarg} to successive refinement.

The main result of the paper can now be stated as follows.

\begin{thm} [Parametric representation]
 \label{thm:parametric}
 Assume that $(d_1, d_2, R_1) \in \Omega$.  
The boundary of the rate-distortion region of successive refinement can be represented as
\begin{align}
 &~ R_2(d_1, d_2, R_1) 
 \two{\notag \\}{}
 =
 \two{&~}{}
  \max
  \left\{\E{ \log  {\frac{1}{\beta_1(X)^{1+\nu_1}}} } - \lambda_1 d_1 - \lambda_2 d_2 - \nu_1 R_1\right\},  \two{\notag \\}{}
 \label{eq:parametric}
\end{align}
where the maximization is over $(\beta_1(x), \nu_1, \lambda_1, \lambda_2) \geq 0$ satisfying, for some $\beta_2(x|y_1) \geq 0$, the constraints
\begin{align}
 \Sigma_2(y_1, y_2)  &\leq  1  \label{eq:sigma12}, \\
\Sigma_1(y_1)  &\leq 1 \label{eq:sigma1}
\end{align}
for all $ (y_1, y_2) \in \mathcal Y_1 \times \mathcal Y_2$.

Furthermore, 
in order for $(P_{Y^*_1|X}, P_{Y^*_2 | X Y_1^*})$ to achieve the infimum in \eqref{eq:RRdsr}, it is necessary and sufficient that 
\begin{align}
\frac{d P_{ X | Y_1^\star = y_1}}{dP_{X}}(x) &=  \frac{\exp \left(-\frac{\lambda_1^\star}{1+\nu_1^\star} \sd_1(x, y_1) \right)}{\beta_1(x)\beta_2(x| y_1)^{-\frac{1}{1+\nu_1^\star}}}    \label{eq:Y1star},\\
\frac{d P_{ X| Y_1^\star = y_1, Y_2^\star = y_2 }}{dP_{X | Y_1^\star  = y_1}}(x) &=  \frac {\exp (-\lambda_2^\star \sd_2(x, y_2))} {\beta_2(x| y_1)}    \label{eq:Y2star},
\end{align}
where
\begin{align}
(\lambda_1^\star, \lambda_2^\star, \nu_1^\star) &= - \nabla R_2 (d_1, d_2, R_1) \label{eq:slopes}, 
\end{align}
and $0 \leq \beta_1(x) \leq 1$, $0 \leq \beta_2(x|y_1) \leq 1$ satisfy
\begin{align}
 \Sigma_2(y_1, y_2)  &\leq   \Sigma_1(y_1) \label{eq:sigma12a} \\
 &\leq 1 \label{eq:sigma1a}.
\end{align}
for all $ (y_1, y_2) \in \mathcal Y_1 \times \mathcal Y_2$.
Finally, the choice 
\begin{align}
 \beta_1^\star(x) &= \E{\beta_2(x| Y_1^\star)^\frac{1}{1+\nu_1^\star}\exp \left(-\frac{\lambda_1^\star}{1+\nu_1^\star} \sd_1(x, Y_1^\star)    \right)}\label{eq:g1star}, \\
 \beta_2^\star(x| y_1) &= 
 \E{\exp(- \lambda_2^\star  \sd_2(x, Y_2^\star)) | Y_1^\star = y_1}.
  \label{eq:g2star} 
\end{align}
satisfies \eqref{eq:Y1star}, \eqref{eq:Y2star}, \eqref{eq:sigma12a}, \eqref{eq:sigma1a} and thus
achieves the maximum in \eqref{eq:parametric}.  Equality in \eqref{eq:sigma1a} is attained for $P_{Y_1^\star}$-a.e. $y_1$, and equality in \eqref{eq:sigma12a} is attained for $P_{Y_2^\star | Y_1^\star = y_1}$-a.e. $y_2$.\footnote{By the definition of a transition probability kernel, the transition probability kernels $P_{Y_2^\star | X, Y_1 = y_1}$ and $P_{Y_2^\star | Y_1^\star = y_1}$ are well defined at every $y_1$ (and not only at $P_{Y_1^\star}$-a.e. $y_1$). }  
\end{thm}

If the source is successively refinable at $(d_1,d_2)$, then the optimal choice is
\begin{align}
\beta_1(x) &= \alpha_1^\star(x)^{\frac{\nu_1}{1+\nu_1}} \alpha_2^\star(x)^{\frac{1}{1+\nu_1}} \label{eq:g1lb},\\
\beta_2(x|y_1) &= \exp\left( - \frac{\lambda_1}{\nu_1}  \sd_1(x, y_1)\right) \alpha_1^{ \star-1}(x) \alpha_2^\star(x) \label{eq:g2lb},\\
\lambda_1 &= - \nu_1 R_{\sd_1}^\prime(d_1),\\
\lambda_2 &= - R_{\sd_2}^\prime(d_2). \label{eq:l2sr}
\end{align}
for an arbitrary $\nu_1>0$, 
where $\alpha_1^\star(\cdot)$, $\alpha_2^\star(\cdot)$ achieve the maximum of \eqref{eq:RR(d)csiszar} for $\{\sd_1, d_1\}$ and $\{\sd_2, d_2\}$, respectively. It is easy to verify that in this case, \eqref{eq:sigma12} and \eqref{eq:sigma1} are satisfied, and the function in \eqref{eq:parametric} equals $R_{\sd_2}(d_2)$ when $R_1=R_{\sd_1}(d_1)$. Plugging \eqref{eq:g1lb}, \eqref{eq:g2lb} into \eqref{eq:Y1star}, \eqref{eq:Y2star} yields the optimal kernels 
 \begin{align}
d P_{ X | Y_1^\star = y_1} (x) &=  \frac{\exp \left(- \frac {\lambda_1}{\nu_1} \sd_1(x, y_1) \right)}{ \alpha_1^\star(x)} dP_{X}(x),\label{eq:sr1} \\
d P_{ X| Y_1^\star = y_1, Y_2^\star = y_2 }  &=  
\frac {\exp (-\lambda_2 \sd_2(x, y_2))} {\alpha_2^\star(x)} \frac{ \alpha_1^\star(x)\, dP_{X | Y_1^\star  = y_1}  } {\exp \left(- \frac {\lambda_1}{\nu_1} \sd_1(x, y_1) \right)} \label{eq:sr2}\\
&= \frac {\exp (-\lambda_2 \sd_2(x, y_2))} {\alpha_2^\star(x)} dP_{X} (x), \label{eq:sr2a}
\end{align}
which coincide with the kernels that achieve the single-stage rate-distortion function \eqref{eq:PYcondstar}, indicating successive refinability. The intuition is as follows. After the first stage of successive refinement is complete, the effective source distribution to be compressed is $P_{ X | Y_1^\star}$. Due to \eqref{eq:sr2a}, the Markov chain condition $P_{ X| Y_1^\star, Y_2^\star} = P_{ X| Y_2^\star}$ holds, where $P_{ X| Y_2^\star}$ is the backward transition probability kernel that achieves the rate-distortion function at $d_2$ for $P_X$. Thus after the second stage the effective source distribution coincides with that of the optimal single-stage rate-distortion code, $P_{ X | Y_2^\star}$. The calculation \eqref{eq:sr2a} also recovers the Markovian characterization of successive refinability due to Equitz and Cover \cite[Th. 2]{equitz1991successive}.

 \thmref{thm:parametric} refines a prior finite alphabet result by Tuncel and Rose~\cite[Th. 4]{tuncel2003computation} and extends it to abstract probability spaces. In the finite alphabet case, the optimality conditions \eqref{eq:g1star}, \eqref{eq:g2star} and \eqref{eq:sigma12a}, \eqref{eq:sigma1a}  were stated in \cite[eq. (47), eq. (46) and eq. (50)]{tuncel2003computation}, respectively. 
The dual representation of the rate-distortion region as a maximum over functions in \eqref{eq:parametric} is new. One reason why such a representation is useful is that by choosing $\beta_1$ and $\beta_2$ appropriately, one can generate outer bounds to the rate-distortion region. For example, choosing $\beta_1$ and $\beta_2$ as in \eqref{eq:g1lb} and \eqref{eq:g2lb} leads to an outer bound to the rate-distortion region in \eqref{eq:reg}, even if the source is not successively refinable. This particular choice also  leads to a nonasymptotic converse bound in \corref{cor:c} in \secref{sec:nonasymptotic} below.

\section{Nonasymptotic converse bound}
\label{sec:nonasymptotic}

We focus on excess distortion codes for successive refinement, that we formally define as follows. 
An $(M_1, M_2, d_1, d_2, \epsilon_1, \epsilon_2)$ code for $(P_X, \sd_1, \sd_2)$ is a pair of encoders $(\sf_1, \sf_2)$ \eqref{eq:enc1}, \eqref{eq:enc2} and decoders $(\sg_1, \sg_2)$ \eqref{eq:dec1}, \eqref{eq:dec2} such that
\begin{align}
 \Prob{\mathcal A_1^c} &\leq \epsilon_1, \label{eq:A1prob}\\
 \Prob{\mathcal  A_2^c} &\leq \epsilon_2,
\end{align}
where $\mathcal A_1$ and $\mathcal A_2$ denote the successful decoding events at first and second stages, respectively:
\begin{align}
 \mathcal A_1 &\triangleq \left\{  \sd_1(X, Y_1) \leq d_1 \right\}, \label{eq:A1}\\
\mathcal  A_2 &\triangleq \mathcal A_1 \cap \left\{  \sd_2(X, Y_2) \leq d_2 \right\} \label{eq:A2},
\end{align}
where $Y_1 = \sg_1(\sf_1(X))$ and $Y_2 = \sg_2 (\sf_1(X), \sf_2(X) )) $. We allow randomized encoders and decoders, in which case $\sf_1,\sf_2,\sg_1,\sg_2$ are transition probability kernels rather than deterministic mappings.

It was shown in \cite{kostina2011fixed} that for single stage compression, the random variable called tilted information, defined in \eqref{eq:dtilted}, plays the key role in the corresponding nonasymptotic fundamental limits. Leveraging the result of \thmref{thm:parametric}, we can define the tilted information for successive refinement as follows. 
\begin{defn}
Fix $P_X$, $\sd_1$, $\sd_2$. Tilted information for successive refinement of $x$ at $(d_1, d_2, R_1) \in \Omega$  is defined as
\begin{equation}
\jmath(x, d_1, d_2, R_1) \triangleq  (1 + \nu_1^\star)\log \frac 1 {\beta_1^\star(x)} - \lambda_1^\star d_1 - \lambda_2^\star d_2 - \nu_1^\star R_1
\end{equation}
where $(\beta_1^\star(\cdot)$, $\lambda_1^\star, \lambda_2^\star, \nu_1^\star)$ achieve the maximum in \eqref{eq:parametric}. 
\label{defn:tilted}
\end{defn}

If the source is successively refinable at $d_1, d_2$, then the tilted information for successive refinement coincides with the tilted information for single stage compression: 
\begin{equation}
\jmath(x, d_1, d_2, R_{\sd_1}(d_1)) = \jmath_{\sd_2}(x, d_2).
\end{equation}

Fixing $(\beta_1(\cdot)$, $\lambda_1, \lambda_2, \nu_1) \geq 0$ that satisfy \eqref{eq:sigma12} and \eqref{eq:sigma1}, for some $\beta_2(\cdot | \cdot) \geq 0$, the notion of tilted information can be generalized by defining
\begin{align}
 F &\triangleq (1 + \nu_1)\log \frac 1 {\beta_1(x)} - \lambda_1 d_1 - \lambda_2 d_2 - \nu_1 \log M_1 \label{eq:F}.
\end{align}
Choosing $(\beta_1(\cdot)$, $\lambda_1, \lambda_2, \nu_1)= (\beta_1^\star(\cdot)$, $\lambda_1^\star,\, \lambda_2^\star,\, \nu_1^\star)$ as in Definition \ref{defn:tilted} would result in $F = \jmath(X, d_1, d_2, \log M_1)$.
 For a given $(M_1, M_2, d_1, d_2, \epsilon_1, \epsilon_2)$ code $(\sf_1, \sf_2, \sg_1, \sg_2)$ with $Y_1 = \sg_1(\sf_1(X))$, it is instructive to split \eqref{eq:F} into two terms (corresponding to both stages of successive refinement):
\begin{align}
F = \nu_1 (F_1  - \log M_1) + F_2,
\end{align}
where
\begin{align}
F_1 &\triangleq \log \frac{\beta_2(X|Y_1)^{\frac{1}{1+\nu_1}}} {\beta_1(X)}  - \frac{\lambda_1}{1+\nu_1} d_1,\\
F_2 &\triangleq \log \frac{\beta_2(X |Y_1)^{-\frac{\nu_1}{1+\nu_1}}}{\beta_1(X)}  - \frac{\lambda_1}{1+\nu_1} d_1 - \lambda_2 d_2. 
\end{align}
Roughly speaking, $F_1$ and $F_2$ represent the estimates of the number of bits about $X$ that need be conveyed at the end of first and second stages in order to satisfy the constraints $\sd_1(X, Y_1) \leq d_1$ and $\sd_1(X, Y_2) \leq d_2$, respectively, i.e. the information content of $X$ relevant to satisfying these constraints. Since we are looking at a fixed rate scenario, and  $F_1$ and $F_2$ are random variables, we expect the excess distortion event to occur once the information contents $F_1$ and $F_2$  exceed those chosen fixed rates. This intuition is made rigorous in the next result, which states that the probability that $F_1, F_2$ are too high for the chosen rates yet the decoding is performed correctly is low.  

\begin{thm}
For an $(M_1, M_2, d_1, d_2, \epsilon_1, \epsilon_2)$ code to exist, it is necessary that for all $(\gamma_1, \gamma_2) > 0$,
\begin{align}
 \Prob{ \{F_1 \geq \log M_1 + \gamma_1\} \cap \mathcal A_1} &\leq \exp(-\gamma_1) \label{eq:c1},\\
  \Prob{ \{F_2 \geq \log M_2 + \gamma_2\} \cap \mathcal A_2} &\leq \exp(-\gamma_2) \label{eq:c2},
\end{align}
where $\mathcal A_1, \mathcal A_2$ are the successful decoding events \eqref{eq:A1}, \eqref{eq:A2}. 
\label{thm:c}
\end{thm}

\begin{proof}[Proof of \thmref{thm:c}]
We employ \thmref{thm:parametric} similar to how \thmref{thm:csiszarg} was employed in the proof of \cite[Th. 7]{kostina2011fixed}.

Let the two-stage encoder and decoder be the random transformations $(P_{W_1|X}, P_{W_2 | X, W_1})$ and $(P_{Y_1|W_1}, P_{Y_2 | W_1, W_2})$, where $W_1$ takes values in $\{1, \ldots, M_1\}$, and $W_2$ takes values in $\left\{1, \ldots, \left \lfloor {M_2}/{M_1} \right \rfloor \right\}$.  

Furthermore, introduce the auxiliary distribution $Q_{W_1 W_2}$, equiprobable on $\{1, \ldots, M_1\} \times \{1, \ldots, \left \lfloor {M_2}/{M_1} \right \rfloor\}$, and let $Q_{Y_1 Y_2}$ be the distribution on $\mathcal Y_1 \times \mathcal Y_2$ that arises after $Q_{W_1 W_2}$ is passed through the random transformation  $P_{Y_1 Y_2|W_1, W_2}$ defined by our code, i.e. $Q_{W_1 W_2} \to  P_{Y_1 Y_2|W_1, W_2} \to Q_{Y_1 Y_2}$. 

To show \eqref{eq:c1}, write, for any $\gamma_1 \geq 0$
\begin{align}
&~ \Prob{ \{F_1 \geq \log M_1 + \gamma_1\} \cap \mathcal A_1}  \notag\\
\leq&~ \int_{x \in \mathcal X} dP_X(x) \sum_{w = 1}^{M_1} P_{W_1|X = x}(w) \int dP_{Y_1|W_1 = w}(y)  
\two{\notag\\ \phantom{=}&~\cdot}{}
\1{ M_1 \leq  \exp\left(  F_1  -\gamma_1 \right) } \1{\sd_1(x, y_1) \leq d_1 } \\
\leq&~  \exp\left(-\gamma_1\right)
  \mathbb E_{P_X \times Q_{Y_1}} \left[ \exp\left( F_1 \right)  \1{\sd_1(X, Y_1)\leq d_1} \right] \label{eq_-Ca}\\
\leq&~   \exp\left(-\gamma_1\right)  \mathbb E_{Q_{Y_1}} [\Sigma_1(Y_1)]
 \label{eq_-Ca1}\\
\leq&~ \exp\left(-\gamma_1\right) \label{eq_-Cb}
\end{align}
where
\begin{itemize}
 \item \eqref{eq_-Ca} follows by upper-bounding $P_{W_1|X = x}(w) \leq 1$, and
\begin{align}
1\left\{ M_1 \leq \exp\left(F_1 -\gamma_1 \right) \right\} 
  \leq&~\frac {\exp\left( -\gamma_1\right)} {M_1} \exp\left(F_1 \right);
\end{align}
 \item \eqref{eq_-Cb} is due to \eqref{eq:sigma1}.
\end{itemize}

We proceed to show \eqref{eq:c2}.
We have, for any $\gamma_2 \geq 0$
\begin{align}
&~  \Prob{ \{F_2 \geq \log M_2 + \gamma_2\} \cap \mathcal A_2}\notag \\
\leq&~ \int_{x \in \mathcal X} dP_X(x) \sum_{w_1, w_2} P_{W_1, W_2|X = x}(w_1, w_2) \notag\\
\phantom{=} &~ \cdot \mathbb E%
\big[ \1{M_2 \leq  \exp\left(  F_2 -\gamma_2 \right) } 
\two{\notag\\ &~ \phantom{\mathbb E \big[ ~} \cdot}{}
\1{\mathcal A_2} | X = x, W_1 = w_1, W_2 = w_2\big]\\
\leq&~  \exp\left(-\gamma_2\right)
 \mathbb E_{P_X \times Q_{Y_1 Y_2}} \left[ \exp\left( F_2 \right)  \1{\mathcal A_2} \right] \label{eq:Ca2}\\
\leq&~  \exp\left(-\gamma_2\right) 
 \mathbb E_{Q_{Y_1 Y_2}} [\Sigma_2(Y_1, Y_2)]
 \label{eq_-Ca1}\\
\leq&~ \exp\left(-\gamma_2\right) \label{eq:Cb2},
\end{align}
where
\begin{itemize}
 \item \eqref{eq:Ca2} follows by upper-bounding 
 $P_{W_1, W_2|X = x}(w_1, w_2) \leq 1$, and 
\begin{align}
 1\left\{ M_2 \leq \exp\left(F_2 -\gamma_2 \right) \right\} \leq \frac {\exp\left( -\gamma_2\right)} {M_2} \exp\left(F_2 \right);
\end{align}
 \item \eqref{eq:Cb2} is due to \eqref{eq:sigma12}.
\end{itemize}

\end{proof}
\thmref{thm:c} immediately leads to the following converse:  for an $(M_1, M_2, d_1, d_2, \epsilon_1, \epsilon_2)$ code to exist, it is necessary that for all $(\gamma_1, \gamma_2) > 0$, 
\begin{align}
\epsilon_1 &\geq  \Prob{ F_1 \geq \log M_1 + \gamma_1 } - \exp(-\gamma_1) \label{eq:c1}, \\
\epsilon_2 &\geq  \Prob{ F_2 \geq \log M_2 + \gamma_2 } - \exp(-\gamma_2) \label{eq:c2}.
\end{align}

In general, $F_1$ and $F_2$ are functions of a given code, which limits the computability of the basic converse in \thmref{thm:c} or that in \eqref{eq:c1}, \eqref{eq:c2}. Fortunately, via elementary probability rules, \thmref{thm:c} immediately leads to a series of corollaries that are computable and useful in several applications as explained below.

The following corollary to \thmref{thm:c} is immediate from the observation that \eqref{eq:g1lb}--\eqref{eq:l2sr} satisfy \eqref{eq:sigma12} and \eqref{eq:sigma1}, and thus $F_1 = \jmath_{\sd_1}(X, d_1)$ and $F_2 = \jmath_{\sd_2}(X, d_2)$ is a valid choice for these functions. 
\begin{cor}
Fix an  $(M_1, M_2, d_1, d_2, \epsilon_1, \epsilon_2)$ code. Then, for all $(\gamma_1, \gamma_2) > 0$, it holds that
\begin{align}
\epsilon_1 &\geq \mathbb P \bigg[  \jmath_{\sd_1}(X, d_1) \geq \log M_1 + \gamma_1\bigg] - \exp(-\gamma_1)\\
\epsilon_2 &\geq \mathbb P \bigg[ \jmath_{\sd_2}(X, d_2) \geq \log M_2 + \gamma_2 \bigg] - \exp(-\gamma_2).
\end{align}
where $\jmath_{\sd_1}$ and $\jmath_{\sd_2}$ are the $\sd_1$- and $\sd_2$-tilted informations (defined in \eqref{eq:dtilted}), respectively. 
\label{cor:c}
\end{cor}

\corref{cor:c} applies whether or not the source  is successively refinable. 

The next corollary recombines the $F_1$ and $F_2$ events in \thmref{thm:c} to yield a bound on the joint error probability $\epsilon_2$ in terms of 
$F$ and $F_1$. This is useful when $F_1$ is a function of $X$ only; for example when $F_1 = \jmath_{\sd_1}(X, d_1)$. 
\begin{cor}
Fix an  $(M_1, M_2, d_1, d_2, \epsilon_1, \epsilon_2)$ code. For all $(\gamma_1, \gamma_2) > 0$, both \eqref{eq:c1} and
\begin{align}
 \epsilon_2 \geq&~ \mathbb P \bigg[ \{F 
  \geq \log M_2 + \nu_1 \gamma_1 + \gamma_2 \} \cup \{F_1 \geq \log M_1 + \gamma_1 \} \bigg] 
  \two{\notag\\&~}{}
  - \exp(-\gamma_1) - \exp(-\gamma_2) \label{eq:cgen}
\end{align}
must hold. 
\label{cor:cg}
\end{cor}

\begin{proof}
 Consider the event
\begin{equation}
\mathcal B \triangleq \left\{  F \geq \log M_2 + \nu_1 \gamma_1 + \gamma_2 \right\} \cup \left\{ F_1 \geq \log M_1 + \gamma_1 \right\}. 
\end{equation}
 Using elementary probability laws and \thmref{thm:c}, write
 \begin{align}
\Prob{\mathcal B} 
=&~ \Prob{\mathcal B  \cap \mathcal A_2^c} + \Prob{\mathcal B  \cap \mathcal A_2 }\\
\leq &~ \epsilon 
+ \Prob{\mathcal B  \cap \mathcal A_2 \cap \{ F_1 \geq \log M_1 + \gamma_1 \} } 
\two{\notag \\ &~ \phantom{ \epsilon } }{}
+  \Prob{\mathcal B  \cap \mathcal A_2 \cap \{ F_1 < \log M_1 + \gamma_1 \}} \\
\leq &~ \epsilon + \exp\left(-\gamma_1\right) + \Prob{ \mathcal A_2 \cap \{ F_2 \geq \log M_2 + \gamma_2 \} }\\
\leq &~ \epsilon + \exp(-\gamma_1) + \exp(-\gamma_2).
\end{align}
\end{proof}

In general, $F_1$ is a function of a given code, which limits the computability of the converse in \corref{cor:cg}. However, when operating at first stage rate close to $R_{\sd_1}(d_1)$, which corresponds to the vertical asymptote in \figref{fig:reg}, $F_1$ becomes a function of~$X$ only, and \eqref{eq:cgen} gives a computable bound that is tighter than \eqref{eq:c}. Indeed, letting $(\beta_1(\cdot)$, $\lambda_1, \lambda_2, \nu_1)$ to achieve the maximum in \eqref{eq:parametric} at $(d_1, d_2, R_{\sd_1}(d_1))$, we obtain $F_1 = \jmath_{\sd_1}(X, d_1)$, which is a function of $X$ only, and $F = \jmath(X, d_1, d_2, R_{\sd_1}(d_1)) + \nu_1 R_{\sd_1}(d_1) - \nu_1 \log M_1$.

Omitting the $F_1$ event from the probability in \eqref{eq:cgen} and choosing $(\beta_1(\cdot)$, $\lambda_1, \lambda_2, \nu_1)$ as in Definition \ref{defn:tilted} so that $F = \jmath(X, d_1, d_2, \log M_1)$, we obtain a bound on the joint error probability $\epsilon_2$ in terms of tilted information only, stated in \corref{cor:c2} below. This is nice because it generalizes the corresponding result for one stage compression \cite[Th. 7]{kostina2011fixed}, and because it leads to a tight second-order result, as explained at the end of this section. 

\begin{cor}
For any  $(M_1, M_2, d_1, d_2, \epsilon_1, \epsilon_2)$ code and for all $(\gamma_1, \gamma_2) > 0$, it holds that
\begin{align}
 \epsilon_2 \geq \mathbb P \bigg[ &\jmath(X, d_1, d_2, \log M_1) 
  \geq \log M_2 + \nu_1^\star \gamma_1 + \gamma_2 %
   \bigg] 
   \two{\notag\\ &}{}
  - \exp(-\gamma_1) - \exp(-\gamma_2) \label{eq:c}
\end{align}
\label{cor:c2}
\end{cor}

In a typical application of the bound in \corref{cor:c2}, $\gamma_1$ and $\gamma_2$ will be chosen so that the terms $\nu_1^\star \gamma_1 + \gamma_2$ inside the probability and $\exp(-\gamma_1) + \exp(-\gamma_2)$ outside are both negligible. Thus, \corref{thm:c} establishes that the excess-distortion probability is roughly bounded below by the complementary cdf of tilted information.

For successively refinable finite alphabet sources, No et al. \cite{no2016strong}  found the dispersion of successive refinement. 
The dispersion of non-successively refinable finite alphabet sources was recently computed in \cite{zhou2017successive}. A straightforward second-order analysis (along the lines of \cite[(103)--(106)] {kostina2011fixed}) of the bound in \corref{cor:c} recovers the converse parts of the dispersion results in \cite{no2016strong} and \cite{zhou2017successive}, respectively, and extends them to abstract stationary memoryless sources.  
Specifically, let $\Qinv{\epsilon}$ be the inverse of the standard Gaussian complementary cdf and let $\mathbf Q^{-1}(\epsilon, \boldsymbol \Sigma)$ be the $K$-dimensional analogue of that function for a Gaussian random vector with zero mean and covariance matrix $\boldsymbol \Sigma$, i.e. $\mathbf Q^{-1}(\epsilon, \boldsymbol \Sigma)$ is the boundary of the set  
\begin{align}
 \left\{ z \in \mathbb R^K \colon \Prob{\mathcal N (\mathbf 0, \boldsymbol \Sigma) \leq z} \geq 1 - \epsilon \right\}.
\end{align}
 
 Consider some $(R_1^\star, R_2^\star)$ on the boundary of the set in \eqref{eq:reg} and some $(L_1^\star, L_2^\star)$ on the boundary of the set
\begin{align}
\Big\{ &(L_1, L_2) \in \mathbb R^2 \colon \two{\notag \\ }{} 
\two{&}{}
\nu_1^\star L_1 + L_2 \leq  \sqrt{\Var{ \jmath(X, d_1, d_2, R_1^\star)} }\, \Qinv{\epsilon_2} \Big\},
\end{align}
where $\nu_1^\star$ is the negative of the derivative of $R_2(d_1, d_2, R_1)$ with respect to $R_1$ at $R_1 = R_1^\star$. An asymptotic analysis of \corref{cor:c2} yields an extension of the converse part of \cite[Th. 11 (i)]{zhou2017successive} to abstract alphabets: 
if an $(M_1, M_2, d_1, d_2, \epsilon_1, \epsilon_2)$ code exists for $n$ i.i.d. copies of $X$, then 
\begin{align}
\log M_1 &\geq n  R_1^\star   + \sqrt n L_1^\star + \bigo{\log n}, \label{eq:M12order}\\
\log M_2 &\geq n R_2^\star + \sqrt n  L_2^\star + \bigo{\log n} \label{eq:M22order}.
\end{align}
When the asymptotic rate at first stage is the vertical asymptote in \figref{fig:reg}, i.e. $R_1^\star = R_{\sd_1}(d_1)$,  then \corref{cor:cg} leads to the following strengthening of \eqref{eq:M12order}, \eqref{eq:M22order}: if an $(M_1, M_2, d_1, d_2, \epsilon_1, \epsilon_2)$ code exists for $n$ i.i.d. copies of $X$, then \eqref{eq:M12order}, \eqref{eq:M22order} hold with $(L_1^\star, \nu^\star L_1^\star + L_2^\star) \in \mathbf Q^{-1}(\epsilon_2, \boldsymbol \Sigma)$ for $\boldsymbol \Sigma$ being the covariance matrix of the two-dimensional random vector $(\jmath_{\sd_1}(X, d_1), \jmath(X, d_1, d_2, R_1^\star))$. The finite alphabet case of this result is the converse part of \cite[Th. 11 (iii)]{zhou2017successive}. The converse result \eqref{eq:M12order}, \eqref{eq:M22order} also holds with $R_1^\star = R_{\sd_1}(d_1)$, $R_2^\star = R_{\sd_2}(d_2)$, and  $\boldsymbol \Sigma$ the covariance matrix of the two-dimensional random vector $(\jmath_{\sd_1}(X, d_1), \jmath_{\sd_2}(X, d_1))$, which is tight if the source is successively refinable \cite[Cor. 13 (iii)]{zhou2017successive}. 

Unlike \cite{zhou2017successive} who focused on the joint probability of error $\epsilon_2$ without placing any further constraint on $\epsilon_1$ apart from the trivial $\epsilon_1 \leq \epsilon_2$, No et al. \cite{no2016strong} considered a formulation that places
separate upper bounds on each of the probabilities that the source is not reproduced within distortion levels $d_1$ and $d_2$, i.e. \eqref{eq:A1prob} and $\Prob{\sd_2(X, Y_2) > d_2} \leq \epsilon_2^\prime$. It is easy to show that \corref{cor:c} continues to hold with $\epsilon_2$ replaced by $\epsilon_2^\prime$. The converse part of \cite[Cor. 6]{no2016strong} then extend it to abstract alphabets as follows:  if an $(M_1, M_2, d_1, d_2, \epsilon_1, \epsilon_2^\prime)$ code under separate error probability formalism exists for $n$ i.i.d. copies of $X$, then 
\begin{align}
\log M_1 &\geq n  R_{\sd_1}(d_1)   + \sqrt{ n \Var{\jmath_{\sd_1}(X, d_1)}} \Qinv{\epsilon_1} \two{\notag \\ &} {} + \bigo{\log n} \label{eq:M12ordera},\\
\log M_2 &\geq n  R_{\sd_2}(d_2)   + \sqrt{ n \Var{\jmath_{\sd_2}(X, d_2)}} \Qinv{\epsilon_2^\prime} \two{\notag \\ &} {} + \bigo{\log n} \label{eq:M22ordera}.
\end{align}

 \section{Proof of \thmref{thm:csiszarg}}
 \label{sec:csiszar}
 In this section, we revisit the beautiful proof of \thmref{thm:csiszarg} by Csisz\'ar \cite{csiszar1974extremum}. We streamline Csisz\'ar's argument by using the the Donsker-Varadhan characterization of the minimum relative entropy, stated below, which will be also instrumental in the proof of \thmref{thm:parametric}.

 \begin{lemma}[{Donsker-Varadhan, \cite[Lemma 2.1]{donsker1975asymptotic}, \cite[Th. 3.5]{polyanskiy2012notes}} ]
Let $\rho \colon \mathcal X \mapsto [-\infty, + \infty]$ and let $\bar X$ be a random variable on $\mathcal X$ such that 
$
\E{\exp\left( - \rho(\bar X)\right) } < \infty
$. 
Then, 
\begin{equation}
  D(X \| \bar X) + \E{ \rho(X) } \geq \log \frac 1 {\E{\exp\left( - \rho(\bar X)\right) }} \label{sc:lemmaverdu}
\end{equation}
with equality if and only if $X$ has distribution $P_{X^\star}$ such that
\begin{equation}
dP_{X^\star}(x) = \frac{\exp\left( -\rho(x)\right)}{\E{\exp\left( - \rho(\bar X)\right) }} dP_{\bar X}(x)  \label{eq:dvXstar}
\end{equation}
\label{lem:dv}
\end{lemma}

We now recall some useful general properties of $R(d)$. 

Fix source distribution $P_X$. For some transition probability kernel $P = P_{Y|X}$, put 
\begin{align}
I(P) &\triangleq I(X; Y)\\
\rho(P) &\triangleq \E{\sd(X, Y)}. 
\end{align}

 \begin{lemma}[{\cite[Lemma 1.1]{csiszar1974extremum}}]
$R(d)$ is non-increasing, convex and 
\begin{equation}
R(d) = \inf_{\substack{P \colon \rho(P) = d}} I(P) \quad d_{\min} < d \leq d_{\max}. \label{eq:eqd}
\end{equation}
\end{lemma}

 Let $F(\lambda)$ denote the maximum of the vertical axis intercepts of the straight lines of slope $-\lambda$ which have no point above the $R(d)$ curve, i.e. using \eqref{eq:eqd} for $\lambda > 0$ (see \figref{fig:reg1}) \footnote{The optimization problem in \eqref{eq:Fl} is known as the Lagrangian dual problem, and the function $F(\lambda)$ as the Lagrange dual.}
\begin{equation}
F(\lambda) \triangleq \inf_P I(P) + \lambda \rho(P). \label{eq:Fl}
\end{equation}
\begin{figure}[htp]
\centering
 \epsfig{file=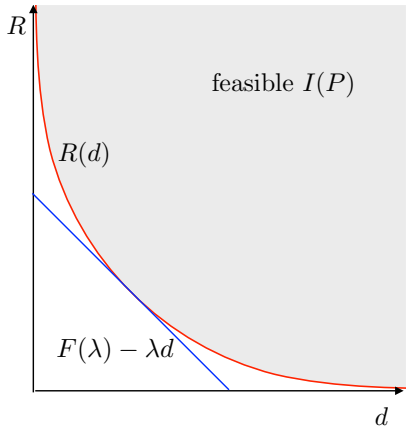,width=\two{.6}{.3}\linewidth}
 \caption[]{Lagrange duality for the rate-distortion problem.} \label{fig:reg1}
\end{figure}
Furthermore, since $R(d)$ is convex and nonincreasing, to each $d \geq d_{\min}$, there exists $\lambda \geq 0$ such that the straight line of slope $-\lambda$ through $(d,  R(d))$ is tangent to the $R(d)$ curve, and
\begin{equation}
R(d) \triangleq \max_{\lambda \geq 0} \left( F(\lambda) - \lambda d \right). \label{eq:lagr1}
\end{equation}

 \thmref{thm:csiszarg} will follow from \eqref{eq:lagr1} and \thmref{thm:nessuf} below. 

\begin{thm}[Necessary and sufficient conditions for an optimizer  \cite{csiszar1974extremum}]
In order for $P_{Y^*|X}$ to achieve the infimum in \eqref{eq:Fl}, it is necessary and sufficient that 
 \begin{equation}
\frac{d P_{X | Y^\star = y}}{dP_{X}}(x)  =  \frac {\exp (-\lambda \sd(x, y)) } {\alpha(x)}  \label{eq:yg}
\end{equation}
where $0 \leq \alpha(x) \leq 1$ satisfies \eqref{eq:csiszarg}.
Furthermore, the choice
\begin{equation}
\alpha^*(x) =  \E{\exp ( -\lambda \sd(x,  Y^*) ) } \label{eq:fstar}
\end{equation}
satisfies \eqref{eq:yg} and \eqref{eq:csiszarg}, and for any $\alpha(x) \geq 0$ satisfying \eqref{eq:csiszarg} we have for all $\tilde P$
\begin{equation}
I(\tilde P) + \lambda \rho(\tilde P) \geq \E{ \log \frac 1 {\alpha(X)}} \label{eq:gmax}
\end{equation}
with equality if and only if $\tilde P$ can be represented as in \eqref{eq:yg}, with the given $\alpha(x)$. 
\label{thm:nessuf}
\end{thm}

\begin{proof}[Proof of \thmref{thm:nessuf}]
Consider the function
\begin{align}
L(P_{Y|X}, P_{\bar Y}) &=  D(P_{Y|X} \| P_{\bar Y} | P_X)  + \lambda \E{\mathsf d(X, Y)} \label{eq:Lbarsc}\\
&= I(X; Y)  + D(Y\|\bar Y) + \lambda \E{\mathsf d(X, Y)} \label{sc:-dtilted1a}\\
&\geq I(X, Y) + \lambda \E{\mathsf d(X, Y)} \label{sc:-dtilted1}
\end{align}

Since equality in \eqref{sc:-dtilted1} holds if and only if $P_{Y} = P_{\bar Y}$, $F(\lambda)$ can be expressed as
\begin{align}
F(\lambda) &=\inf_{P_{\bar Y}} \inf_{P_{Y|X}}   L(P_{Y|X}, P_{\bar Y})   \label{sc:R(d)double}.
\end{align}

Denote 
\begin{equation}
\Sigma_{\bar Y}(x) \triangleq  \E{\exp ( -\lambda \sd(x, \bar Y) ) }.  \label{eq:sigmabarY}
\end{equation}
Since $\sd(x, y) \geq 0$,  we have $ 0 \leq \Sigma_{\bar Y}(x) \leq 1$ 
, and 
\lemref{lem:dv} applies to conclude that equality in 
\begin{equation}
D(P_{Y|X = x} \| P_{\bar Y})+ \lambda \E{\mathsf d(x, Y)|X = x}  \geq \log \frac 1  { \Sigma_{\bar Y}(x)},  \label{sc:-dtilted2}
\end{equation}
is achieved if and only if $P_{Y|X = x} = P_{\bar Y^*| X = x}$, where $P_{\bar Y^*| X = x}$ is determined from 
\begin{equation}
\log \frac{d P_{\bar Y^* | X = x}(y)}{dP_{\bar Y}(y)} + \lambda \sd(x, y) =  \log \frac {1} {\Sigma_{\bar Y}(x) }. \label{eq:tilted}
\end{equation}

Applying \eqref{sc:-dtilted2} to solve for the inner minimizer in \eqref{sc:R(d)double}, we obtain
\begin{align}
F(\lambda) &=   \inf_{P_{\bar Y}} \E{\log \frac 1 { \Sigma_{\bar Y}(X)}}  \label{eq:Rxmaxinf}\\
&= \E{\log \frac 1 { \Sigma_{Y^*}(X)}}  \label{sc:-dtilted4}
\end{align}
where \eqref{sc:-dtilted4} holds by the assumption \eqref{item:b1}. %

Although for a fixed $P_{\bar Y}$ we can always define the tilted distribution  
$P_{\bar Y^* | X}$ via \eqref{eq:tilted}, in general we cannot claim that the marginal distribution $P_{\bar Y^\star}$ that results after applying the random transformation $P_{\bar Y^* | X}$ to $P_X$ coincides with $P_{\bar Y}$. This happens if and only if $P_{\bar Y}$ is such that for $P_{\bar Y}$-a.e. $y$, 
\begin{equation}
\E{ \frac{\exp \left( - \lambda \mathsf d(X, y)\right) } {\Sigma_{\bar Y}(X) }}  = 1 \label{sc:-dtiltedoneeq}.
\end{equation}
Since by the assumption \eqref{item:b1}, these exists $P_{Y^* | X}$ that achieves \eqref{sc:-dtilted4}, condition \eqref{sc:-dtiltedoneeq} must hold for $P_{\bar Y} = P_{Y^*}$. Using this observation together with \eqref{sc:-dtilted4}, we conclude that $P_{Y^* | X}$ in \eqref{eq:yg} with $\alpha(x) = \alpha^*(x)$, where $\alpha^*(x)$ is defined in \eqref{eq:fstar}, 
is necessary and sufficient to achieve the minimum of $F(\lambda)$ in \eqref{eq:Fl}. In particular, \eqref{eq:yg} is a necessary condition for the minimizer. 

We now show that $\alpha^*(x)$ satisfies \eqref{eq:csiszarg}, which implies that both \eqref{eq:yg} and \eqref{eq:csiszarg} are necessary.  
Since $P_X \to P_{Y^* | X} \to P_{Y^*}$,  equality in \eqref{sc:-dtiltedoneeq} particularized to $P_{Y^*}$ holds for $P_{Y^*}$-a.s. $y$, which is equivalent to equality in \eqref{eq:csiszarg}. To show \eqref{eq:csiszarg} for all $y$, note using \eqref{eq:Rxmaxinf} that for any $P_{\bar Y}$, 
\begin{align}
\E{ \log \frac 1 {\Sigma_{ Y^*}(X)} }  &\leq   \E{ \log \frac 1 {\Sigma_{\bar Y}(X)} } \label{sc:-dtilted5}
\end{align}
For an arbitrary $\bar y \in \mathcal Y$ and $0 \leq \epsilon \leq 1$, let
\begin{equation}
 P_{\bar Y} = (1 - \epsilon)P_{Y^*} + \epsilon \delta_{\bar y} \label{eq:pertubY}
\end{equation}
for which 
\begin{equation}
\Sigma_{\bar Y}(x) =  (1 - \epsilon) \Sigma_{Y^*}(x) + \epsilon \exp\left( -\lambda^* \mathsf d(x, \bar y)\right)  \label{sc:-dtilted6}
\end{equation}
Substituting \eqref{sc:-dtilted6} in \eqref{sc:-dtilted5}, we obtain
\begin{align}
 0 &\geq \E{\log \frac {\Sigma_{\bar Y }(X)} {\Sigma_{ Y^*}(X)}  }  \label{sc:-dtilted7a}\\
 &=  \E{ \log \left[ 1 - \epsilon + \epsilon \frac{\exp\left( -\lambda \mathsf d(X, \bar y)\right) }{\Sigma_{ Y^* }(X)} \right] } \label{sc:-dtilted7}\\
 &= \log \left( 1 - \epsilon \right) + \E{\log \left[ 1 + \frac{\epsilon}{1 - \epsilon}  \frac{\exp\left( -\lambda \mathsf d(X, \bar y)\right) }{\Sigma_{ Y^* }(X)}  \right] } \label{sc:-dtilted8}
\end{align}
Since the difference quotient of the second term satisfies $0 \leq \frac 1 \epsilon \log \left(1 + \frac{\epsilon x}{1 - \epsilon} \right) \leq \frac{x \log e}{1 - \delta}$ for all $x \geq 0$ and $0 \leq \epsilon \leq \delta \leq 1$, by the dominated convergence theorem, the right derivative of \eqref{sc:-dtilted8} with respect to $\epsilon$ evaluated at $\epsilon = 0$ is
\begin{equation}
 \E{- 1+  \frac{\exp\left( -\lambda \mathsf d(X, \bar y)\right) }{\Sigma_{ Y^*}(X)}  } \log e \leq 0 
\end{equation}
where the inequality holds because otherwise \eqref{sc:-dtilted7a} would be violated for sufficiently small $\epsilon$. This concludes the proof that $\alpha^*(x)$ in \eqref{eq:fstar} satisfies  condition \eqref{eq:csiszarg}, so both \eqref{eq:yg} and \eqref{eq:csiszarg} in Theorem \ref{thm:nessuf} are necessary. 

The sufficiency of \eqref{eq:yg} and \eqref{eq:csiszarg} for $P_{Y^*|X}$ to achieve the minimum in \eqref{eq:Fl} follows from \eqref{eq:gmax}. 
To show \eqref{eq:gmax}, fix any $\alpha(x)$ satisfying \eqref{eq:csiszarg} and use the concavity of the logarithm to show that
\begin{align}
\two{&~}{}I(\tilde P) + \lambda \rho(\tilde P) \two{\notag\\}{}
\geq &~\E{\log \frac 1 {\Sigma_{Y^*}(X) }}  \\
=&~  \E{\log \frac 1 {\alpha(X) }} -  \E{\log \frac{\Sigma_{Y^*}(X)}{\alpha(X)}}\\
\geq&~ \E{\log \frac 1 {\alpha(X) }} - \log \EE_{P_X \times P_{Y^*}}{ \frac{\exp ( -\lambda \sd(X, Y^*) )}{\alpha(X)} } \label{eq:eqcond1}\\
\geq&~ \E{\log \frac 1 {\alpha(X) }} \label{eq:eqcond2}.
\end{align}
For the equality condition, observe that strict concavity of logarithm implies that  equality in \eqref{eq:eqcond1} holds if and only if the ratio $\frac{\Sigma_{Y^*}(X)}{\alpha(X)}$ is constant, while equality in \eqref{eq:eqcond2} holds if and only if that constant is $1$. 
 \end{proof}

 \section{Proof of \thmref{thm:parametric}}
 \label{sec:proof}
 
 Fix the source distribution $P_X$. For a transition probability kernel $P = P_{Y_1 Y_2|X}$, put 
\begin{align}
I_1(P) &\triangleq I(X; Y_1)\\
I_2(P) &\triangleq I(X; Y_1, Y_2)\\
\rho_1(P) &\triangleq \E{\sd_1(X, Y_1)}\\
\rho_2(P) &\triangleq  \E{\sd_2(X, Y_2)}
\end{align}

  \begin{lemma}
The function $R_2(d_1, d_2, R_1)$ is non-increasing as a function of each argument when the others are kept fixed, jointly convex, and 
\begin{equation}
R_2(d_1, d_2, R_1) = \inf_{\substack{P \colon \\
I_1(P) = R_1\\
\rho_1(P) = d_1\\
\rho_2(P) = d_2\\
}} I_2(P), \quad (d_1, d_2, R_1)\in \Omega \label{eq:eqsr},
\end{equation}
where the set $\Omega$ is defined in \eqref{eq:omega}.
\label{lem:convex}
\end{lemma}
\begin{proof}®
That $R_2(d_1, d_2, R_1)$ is non-increasing is obvious by definition. To show convexity, note first that  since $u \log \frac u v$ is a convex function of $(u, v)$, $D(P \| Q)$ is a convex function of $(P, Q)$, and so $I_2(P)$ is a convex function of $P$.  Therefore, $I_1(P)$ is convex as a composition of a convex function $I_2(P)$ with  an affine mapping $P \mapsto P_{Y_1 | X}$. 

Let the probability kernel $P_a$ attain $R_2( d_1^a,  d_2^a,  R_1^a)$ and $P_{b}$ attain  $R_2(d_1^b, d_2^b, R_1^b)$. Let 
$P = \epsilon P_{a} + (1 - \epsilon) P_{b}$, $d_1 = \epsilon d_1^a + (1 - \epsilon) d_1^b$, $d_2 = \epsilon d_2^a + (1 - \epsilon) d_2^b$, $R_1 = \epsilon R_1^a + (1 - \epsilon) R_1^b$. 
Since $I_1(P)$ is convex and $ \rho_1(P)$, $\rho_2(P)$ are affine,
\begin{align}
I_1(P) &\leq R_1 \label{eq:constr1}\\
 \rho_1(P) &\leq d_1\\
 \rho_2(P) &\leq d_2. \label{eq:constr3}
\end{align}
Furthermore, by convexity of $I_2(P)$,
\begin{align}
I_2(P) &\leq \epsilon I_2(P_a) + (1 - \epsilon)   I_2(P_b)\\
&= \epsilon R_2( d_1^a,  d_2^a, R_1^a)  \label{eq:Iconv}
\two{\\
&\phantom{=} }{} + (1 - \epsilon) R_2(d_1^b, d_2^b, R_1^b). \two{\notag}{}
\end{align}
Convexity of $R_2(d_1, d_2, R_1)$ follows by minimizing the left side of \eqref{eq:Iconv} over $P$ satisfying the constraints \eqref{eq:constr1}--\eqref{eq:constr3}. To show \eqref{eq:eqsr}, rewrite  $ R_2(d_1, d_2, R_1)$ as
\begin{align}
  R_2(d_1, d_2, R_1) = \inf_{\substack{\tilde R_1 \leq R\\
\tilde d_1 \leq d_1\\
\tilde d_2 \leq d_2
}} \tilde R_2(\tilde d_1, \tilde d_2, \tilde R_1), \label{eq:R2zerodual}
\end{align}
where $\tilde R_2(\cdot, \cdot, \cdot)$ denotes the function in the right side of \eqref{eq:eqsr}. Since for $(d_1, d_2, R_1) \in \Omega$, the function $R_2(d_1, d_2, R_1)$ is strictly decreasing in all arguments, the infimum \eqref{eq:R2zerodual} is achieved at the boundary, and \eqref{eq:eqsr} follows.
\end{proof}

 Put 
\begin{align}
F(\nu_1, \lambda_1, \lambda_2) \triangleq \max \{  h \geq 0 \colon \forall (d_1, d_2, R_2) &
 \two{\notag\\}{\quad} h - \lambda_1 d_1 - \lambda_2 d_2 - \nu_1 R_1 \leq R_2(d_1, d_2, R_2) \two{&}{}\},
\end{align}
i.e. $F(\nu_1, \lambda_1, \lambda_2)$ is the maximum of the $R_2$ axis intercepts of the hyperplanes $h - \lambda_1 d_1 - \lambda_2 d_2 - \nu_1 R_1$ which have no point inside of the rate-distortion region in \eqref{eq:reg}, i.e. for $(\lambda_1, \lambda_2, \nu_1) \geq 0$ %
\begin{align}
F(\nu_1, \lambda_1, \lambda_2) &\triangleq \inf_{P_{Y_1|X},P_{ Y_2 | XY_1}} L \left( P_{Y_1|X}, P_{Y_2 |X Y_1} \right), \label{eq:Lstar}
\end{align}
where
 \begin{eqnarray*}
L \left( P_{Y_1|X}, P_{Y_2 |X Y_1} \right)  
&\!\!\!\triangleq & \!\!\!I(X; Y_1,Y_2)+ \nu_1 I(X; Y_1)   
\two{\\
&\!\!\!&\!\!\!}{}+ \lambda_1 \E{\sd_1(X, Y_1)} + \lambda_2 \E{\sd_2(X, Y_2)}. \notag
\end{eqnarray*}
In other words, $L \left( P_{Y_1|X}, P_{Y_2 |X Y_1} \right)  $ is the Lagrangian and \eqref{eq:Lstar} is the Lagrangian dual problem. 

Since $R_2(d_1, d_2, R_1)$ is convex and nonincreasing, to each $(d_1, d_2, R_1)$ such that $R_2 (d_1, d_2, R_1) < \infty$, there exists $(\nu_1, \lambda_1, \lambda_2) \geq 0$ such that the hyperplane $h - \lambda_1 d_1 - \lambda_2 d_2 - \nu_1 R_1$ that passes through $(d_1, d_2, R_1)$ is tangent to the surface $R_2(d_1, d_2, R_1)$, and
\begin{align}
 &~ R_2(d_1, d_2, R_1) 
 \two{\notag\\
 =&~}{=} \max_{(\nu_1, \lambda_1, \lambda_2) \geq 0} ( F(\nu_1, \lambda_1, \lambda_2) - \nu_1 R_1-\lambda_1 d_1-\lambda_2 d_2). \label{eq:lagr2}
\end{align}
\thmref{thm:parametric} is an immediate consequence of \eqref{eq:lagr2} and \thmref{thm:nssuccref} below.
\begin{thm}[Necessary and sufficient conditions for an optimizer]
In order for $(P_{Y^*_1|X}, P_{Y^*_2 | X Y_1^*})$ to achieve the infimum in \eqref{eq:Lstar}, it is necessary and sufficient that 
\begin{align}
\frac{d P_{ X | Y_1^* = y_1}}{dP_{X}}(x) &=  \frac{\exp \left(-\frac{\lambda_1}{1+\nu_1} \sd_1(x, y_1) \right)}{\beta_1(x)\beta_2(x| y_1)^{-\frac{1}{1+\nu_1}}}    \label{eq:Y1star0}\\
\frac{d P_{ X| Y_1^* = y_1, Y_2^* = y_2 }}{dP_{X | Y_1^*  = y_1}} &=  \frac {\exp (-\lambda_2 \sd_2(x, y_2))} {\beta_2(x| y_1)},    \label{eq:Y2star0}
\end{align}
where $0 \leq \beta_1(x)\leq 1$, $0 \leq \beta_2(x| y_1)\leq 1$ satisfy \eqref{eq:sigma12a}, \eqref{eq:sigma1a}.

Furthermore, the choice
\begin{align}
 \beta_1^*(x) 
  &= \E{\beta_2^*(x| Y_1^*)^{\frac{1}{1+\nu_1}}\exp \left(-\frac{\lambda_1}{1+\nu_1} \sd_1(x, Y_1^*)  \right)} \label{eq:g1starb}\\
 \beta_2^*(x| y_1) &= 
 \E{\exp(- \lambda_2 \sd_2(x, Y_2^*)) | Y_1^* = y_1} 
  \label{eq:g2starb}
\end{align}
satisfies \eqref{eq:sigma12a}, \eqref{eq:sigma1a}, \eqref{eq:Y1star0}, \eqref{eq:Y2star0}.

Finally, for any $\beta_1(x) \geq 0$, $\beta_2(x| y_1) \geq 0$ satisfying \eqref{eq:sigma12}, \eqref{eq:sigma1},
we have for all $P_{Y_1 Y_2 | X}$
\begin{equation}
L \left( P_{Y_1|X}, P_{Y_2 | XY_1} \right) \geq (1+\nu_1)\E{ \log \frac 1 {\beta_1(X)}} \label{eq:g1max}
\end{equation}
with equality if and only if $(P_{Y_1|X}, P_{Y_2 | XY_1})$ can be represented as in \eqref{eq:Y1star0}, \eqref{eq:Y2star0}, with the given $(\beta_1, \beta_2)$.%
\label{thm:nssuccref}
\end{thm}

\begin{proof}[Proof of Theorem~\ref{thm:nssuccref}]
The proof builds on the groundwork laid out in our proof of \thmref{thm:nessuf}. In the first part of the proof, we will use the Donsker-Varadhan lemma and the assumption of the existence of optimizing kernels to characterize the optimal $\beta_1^*(x)$, $\beta_2^*(x|y_1)$ as well as  $P_{Y_1^* | X}$ and $P_{Y_2^* | XY_1^*}$. We will apply the Donsker-Varadhan lemma twice, first for the second stage and then, thinking of the optimized rate at second stage as modifying the distortion measure at first stage, for the first stage. This reasoning, concluding at \eqref{eq:sigma2eq1} below, will also ensure that equalities in \eqref{eq:sigma12a} and \eqref{eq:sigma1a} hold for $P_{Y_1^* Y_2^*}$-a.e. $(y_1, y_2)$. 

The second part of the proof, \eqref{eq:y1y2ae}--\eqref{eq:der}, shows the necessity of \eqref{eq:sigma12a} and \eqref{eq:sigma1a} for all $(y_1, y_2)$. This involves perturbing $P_{Y_1^* Y_2^*}$ by a delicately chosen auxiliary distribution and using the optimality of $P_{Y_1^* Y_2^*}$ to claim \eqref{eq:sigma12a} and \eqref{eq:sigma1a}. %

Having established these necessary conditions, we will proceed to show their sufficiency in the third and final part of the proof, \eqref{eq:optcond5-}--\eqref{eq:optcond5+}.
 
First, we show that
\begin{align}
\inf_{P_{Y_1|X},\, P_{Y_2 |X Y_1}} L \left( P_{Y_1|X}, P_{Y_2 | X Y_1} \right) = \E{ \log \frac 1 {\beta_1^*(X)^{1+\nu_1}}}.
\end{align}
For fixed probability kernels $P_{\bar Y_1}$ and $P_{\bar Y_2 | \bar Y_1}$, consider the function
\begin{eqnarray}
\lefteqn{\bar{L} \left( P_{Y_1|X}, P_{Y_2 | XY_1}, P_{\bar Y_1}, P_{\bar Y_2 | \bar Y_1} \right)}  \label{eq:Lbar}
\\
  & \!\!\!\triangleq& \!\!\!D(P_{Y_2 | XY_1} \| P_{ \bar Y_2 | \bar Y_1 } | P_{XY_1}) + (1+\nu_1) D(P_{Y_1| X} \| P_{\bar Y_1} | P_X)  
  \two{\notag\\
  &\!\!\!& \!\!\!}{}+ \lambda_1 \E{\sd_1(X, Y_1)} + \lambda_2 \E{\sd_2(X, Y_2)}. \notag
\end{eqnarray}
Since
\begin{align}
\two{&~}{} L \left( P_{Y_1^*|X}, P_{Y_2^* | X Y_1^*} \right) \two{\notag\\}{}
 = &~  \inf_{ P_{Y_1|X},\, P_{Y_2 | XY_1},\, P_{\bar Y_1},\, P_{\bar Y_2 | \bar Y_1}}  \bar{L} \left( P_{Y_1|X}, P_{Y_2 | XY_1},  P_{\bar Y_1}, P_{\bar Y_2 | \bar Y_1} \right)  \\
=&~  \bar{L} \left( P_{Y_1^*|X}, P_{Y_2^* | XY_1^*},  P_{Y_1^*}, P_{Y_2^* | Y_1^*} \right) 
\end{align}
we have
\begin{align}
\bar{L} \left( P_{Y_1|X}, P_{Y_2 | XY_1},  P_{\bar Y_1}, P_{\bar Y_2 | \bar Y_1}  \right) &\geq L \left( P_{Y_1^*|X}, P_{Y_2^* | XY_1^*} \right), \label{eq:Lg}
\end{align}
with equality if and only if $(P_{Y_1|X}, P_{Y_2 | XY_1}, P_{\bar Y_1}, P_{\bar Y_2 | \bar Y_1}) = (P_{Y^*_1|X}, P_{Y^*_2 | XY_1^*}, P_{Y_1^*}, P_{Y_2^* | Y_1^*})$.
Applying Lemma \ref{lem:dv} twice, we compute the minimum of the left side of \eqref{eq:Lg} particularized to $ P_{\bar Y_1} = P_{Y_1^*}$ and $P_{\bar Y_2 | \bar Y_1} = P_{Y_2^* | Y_1^*} $:
\begin{align}
&~\bar{L} \left( P_{Y_1|X}, P_{Y_2 | XY_1},  P_{Y_1^*}, P_{Y_2^* | Y_1^*} \right)  \\
\geq &~ \!\!\!\inf_{P_{Y_1|X}} \bigg\{ (1+\nu_1) D(P_{Y_1 | X} \| P_{Y_1^*} | P_X) + \lambda_1 \E{\sd_1(X, Y_1)} 
 \notag \\
& + \inf_{P_{Y_2 | XY_1}} \left\{  D(P_{Y_2 | XY_1} \| P_{  Y_2^* |  Y_1^* } | P_{XY_1})  + \lambda_2 \E{\sd_2(X, Y_2)}\right\} \!\!\bigg\} \notag\\
 = &~\!\!\!\inf_{P_{Y_1|X}} \bigg\{ (1+\nu_1)D(P_{Y_1 | X} | P_{Y_1^*} | P_X) + \lambda_1 \E{\sd_1(X, Y_1)} 
 \two{\notag\\
&~}{}+ \E{ \log \frac 1 {\beta_2^*(X | Y_1)}} \bigg\}  \\ 
= &~(1+\nu_1)\E{\log \frac{1}{\beta_1^*(X)}} \label{eq:s1}
\end{align}
where $\beta_1^*(x)$ and $\beta_2^*(x|y_1)$ are given in \eqref{eq:g1starb} and \eqref{eq:g2starb}, respectively, and the optimizing  $P_{Y_1^* | X}$ and $P_{Y_2^* | XY_1^*}$ are specified in \eqref{eq:Y1star0} and \eqref{eq:Y2star0}, letting $\beta_1(x) = \beta_1^*(x)$ and $\beta_2(x|y_1) = \beta_2^*(x|y_1)$ therein.

We proceed to show that $\beta_1^*(x)$ and $\beta_2^*(x|y_1)$ satisfy \eqref{eq:sigma12a} and \eqref{eq:sigma1a}. For $P_{Y_1^*}$-a.e. $y_1$, we take expectations with respect to $P_X$ of both sides of \eqref{eq:Y1star} to conclude that
\begin{equation}
\Sigma_1(y_1) = 1 \; . 
\end{equation}
Likewise, for $P_{Y_1^* Y_2^*}$-a.e. $(y_1, y_2)$, we take expectations with respect to $P_{XY_1^*}$ of both sides of \eqref{eq:Y2star} to conclude that 
\begin{equation}
\Sigma_2(y_1, y_2) = 1 \; . \label{eq:sigma2eq1}
\end{equation} 
We next proceed to show that
\begin{equation}
\Sigma_2(y_1, y_2) \leq 1 \quad P_{Y_1^*}\text{-a.e. } y_1, \forall y_2 \in \mathcal Y_2 \label{eq:y1y2ae}
\end{equation}

Particularizing the left side of \eqref{eq:Lg} to $P_{Y_1 | X} = P_{Y_1^* | X}$, $P_{\bar Y_1} = P_{Y_1^*}$, $P_{\bar Y_2 | \bar Y_1} = P_{Y_2^* | Y_1^*}$ we apply Lemma \ref{lem:dv} to characterize the minimum of the left side of \eqref{eq:Lg} as
\begin{eqnarray}
\lefteqn{ \bar{L} \left( P_{Y_1^*|X}, P_{Y_2 | XY_1^*},  P_{Y_1^*}, P_{Y_2^* | Y_1^*} \right)}   \label{eq:optcond0} \\
& \geq& (1+\nu_1)I(X; Y_1^*) + \lambda_1 \E{\sd_1(X, Y_1^*)} 
\two{\notag \\
 &&}{} +
  \inf_{P_{Y_2 | XY_1^*}} \bigg\{  D(P_{Y_2 | X Y_1^*} \| P_{  Y_2^* |  Y_1^* } | P_{XY_1^*})  
  \two{\notag \\
&&  \;\;\;\;\;}{}+ \lambda_2 \E{\sd_2(X, Y_2)}  \bigg\}
\notag
\end{eqnarray}
To evaluate the infimum in \eqref{eq:optcond0}, we apply \thmref{thm:csiszarg} to conclude that for $P_{Y_1^*}$-a.e. $y_1$, it holds that
\begin{align}
&~ D(P_{Y_2 | X,Y_1^* = y_1} \| P_{  Y_2^* |  Y_1^* = y_1 } | P_{X | Y_1^* = y_1}) 
\two{\notag \\
 \phantom{\geq}&~}{}
 + \lambda_2 \E{\sd_2(X, Y_2) | Y_1^* = y_1} 
 \two{\notag\\}{}
 \geq \two{&~}{} \E{ \left. \log \frac 1 {\beta_2^*(X | y_1)} \right| Y_1^* = y_1} 
\end{align}
with
\begin{align}
\E{\frac{\exp(- \lambda_2 \sd_2(X, y_2))}{\beta_2^*(X| y_1)} | Y_1^* = y_1} \leq 1 \;\; ~ P_{Y_1^*}\text{-a.e. } y_1, \forall y_2 
\end{align}
which, using \eqref{eq:Y1star}, is equivalent to \eqref{eq:y1y2ae}. 

To finish the proof of \eqref{eq:sigma12a} and \eqref{eq:sigma1a}, it remains to show that for all $y_1, y_2$ outside of the support of $P_{Y_1^* Y_2^*}$, \eqref{eq:sigma12a} and \eqref{eq:sigma1a} hold. Consider
\begin{align}
&~ \bar{L} \left( P_{Y_1|X}, P_{Y_2 | XY_1},  P_{\bar{Y}_1}, P_{\bar{Y}_2 | \bar{Y}_1} \right) \notag \\
\geq &~ \inf_{P_{Y_1|X}} \bigg\{ (1+\nu_1) D(P_{Y_1 | X} \| P_{\bar{Y}_1} | P_X) + \lambda_1 \E{\sd_1(X, Y_1)} 
 \notag \\
&
+ \inf_{P_{Y_2 | XY_1}} \left\{  D(P_{Y_2 | XY_1} \| P_{  \bar{Y}_2 |  \bar{Y}_1 } | P_{XY_1})  + \lambda_2 \E{\sd_2(X, Y_2)}\right\} \!\!\bigg\} \\
=&~ \inf_{P_{Y_1|X}} \bigg\{ (1+\nu_1) D(P_{Y_1 | X} \| P_{\bar{Y}_1} | P_X) + \lambda_1 \E{\sd_1(X, Y_1)} 
\two{ \notag \\ &}{}
+ \E{ \log \frac 1 {\bar \beta_2(X|Y_1)}} \!\!\bigg\} \\
= &~(1+\nu_1)\E{\log \frac{1}{\bar \beta_1(X)}} \label{eq:g1bar0}
\end{align}
where
\begin{align}
 \bar \beta_1(x) 
   &= \E{\bar{\beta}_2(x| \bar{Y}_1)^{\frac{1}{1+\nu_1}}\exp \left(-\frac{\lambda_1}{1+\nu_1} \sd_1(x, \bar{Y}_1)  \right)} \notag \\
 \bar \beta_2(x| y_1) &= \E{\exp(- \lambda_2 \sd_2(x, \bar Y_2)) | \bar Y_1 = y_1} \; . 
\end{align}
Due to \eqref{eq:Lg},
\begin{align}
\inf_{P_{\bar Y_1}, P_{\bar Y_2 | \bar Y_1}}  \E{\log \frac{1}{\bar \beta_1(X)}} = \E{\log \frac 1 {\beta_1^*(X)}} \label{eq:g1barinf}
\end{align}

Now, we choose $P_{\bar Y_1}$ and $P_{\bar Y_2 | \bar Y_1}$ (not independently of each other!) as
\begin{align}
P_{\bar Y_1 \bar Y_2}(y_1, y_2) = \begin{cases}
(1 - \epsilon) P_{ Y^*_1 Y^*_2}(y_1, y_2)				&\!\!\text{for $P_{Y_1^*}$-a.e. $y_1$}\\
 \epsilon \delta_{y_1} P_{\bar Y_2 | \bar Y_1 = y_1}(y_2) & \!\!\text{otherwise}
\end{cases} 
\end{align}
for some $0 \leq \epsilon \leq 1$, where $P_{\bar Y_2 | \bar Y_1}$ is an arbitrary transition probability kernel.

With this choice, 
\begin{align}
\bar \beta_1(x)  =&  (1 - \epsilon) \beta_1^*(x) \two{\\
\phantom{=}& \notag}{}+ \epsilon  \bar \beta_2 (x|  y_1)^{\frac{1}{1+\nu_1}}\exp \left(-\frac{\lambda_1}{1+\nu_1} \sd_1(x, y_1)  \right)\\
  \bar \beta_1^\prime(x) |_{\epsilon = 0} =& - \beta_1^*(x) \two{\\
  \phantom{=}& \notag}{} + \bar \beta_2 (x|  y_1)^{\frac{1}{1+\nu_1}}\exp \left(-\frac{\lambda_1}{1+\nu_1} \sd_1(x, y_1)  \right)
\end{align}
 Due to \eqref{eq:g1barinf}, the minimum of \eqref{eq:g1bar0} is attained at $\epsilon = 0$, so its right derivative with respect to $\epsilon$ evaluated at  $\epsilon  = 0$ must be nonnegative: 
\begin{align}
\two{&~}{} \left. \frac{\partial}{\partial \epsilon} \E{ \log \frac 1 {\bar \beta_1(X)}  } \right|_{\epsilon = 0} \label{eq:dif0a} 
\two{\\}{}
=&~ 1 - \E{\frac{\bar \beta_2 (X|  y_1)^{\frac{1}{1+\nu_1}}\exp \left(-\frac{\lambda_1}{1+\nu_1} \sd_1(X, y_1)  \right)}{\beta_1^*(X)}} \two{\notag}{} \\
\geq&~ 0, \label{eq:dif0}
  \end{align}
 and  \eqref{eq:sigma1a} follows by substituting $P_{\bar Y_2 |\bar Y_1} = P_{Y_2^* | Y_1^*}$ in \eqref{eq:dif0}. Bringing the differentiation inside of the expectation is permitted by the dominated convergence theorem: the negative of the integrand in \eqref{eq:dif0a} is $\log ((1-\epsilon) a + \epsilon b) = \log (1 - \epsilon) + \log a + \log \left( 1 + \frac {\epsilon}{1 - \epsilon} \frac{b}{a} \right)$, for some $a > 0$, $b > 0$, and the difference quotient of the last term is bounded as $0 \leq \frac 1 \epsilon \log  \left( 1 + \frac {\epsilon}{1 - \epsilon} \frac{b}{a} \right) \leq \frac{b}{a} \frac{\log e}{1 - \delta}$, for all $0 \leq \epsilon \leq \delta < 1$. 
 
  To show \eqref{eq:sigma12a}, notice that \eqref{eq:dif0} implies that the necessary condition for $P_{Y_1^* | X}$, $P_{Y_2^* | X Y_1^*}$ to achieve the minimum is that \eqref{eq:dif0} holds for all choices of the auxiliary kernel $P_{\bar Y_2 | \bar Y_1}$, and so
\begin{align}
 \sup_{P_{\bar Y_2 | \bar Y_1}}\E{\frac{\bar \beta_2 (X|  y_1)^{\frac{1}{1+\nu_1}}\exp \left(-\frac{\lambda_1}{1+\nu_1} \sd_1(X, y_1)  \right)}{\beta_1^*(X)}} \leq 1 \; .\label{eq:dif1}
\end{align}
To simplify \eqref{eq:dif1}, we will find the conditions under which $P_{Y_2^* | Y_1^* = y_1}$ attains the supremum in the left side of \eqref{eq:dif1}. Put
\begin{equation}
P_{\bar Y_2 | \bar Y_1} = (1 - \epsilon) P_{Y_2^* | Y_1^*} + \epsilon \delta_{y_2}.  
\end{equation}
 With this choice, 
\begin{align}
\bar \beta_2(x|y_1) &= (1 - \epsilon) \beta_2^*(x|y_1) + \epsilon \exp(- \lambda_2 \sd_2(x, y_2)) \\
 \frac{\partial}{\partial \epsilon}  \log \bar \beta_2(x|y_1) &~\big|_{\epsilon = 0} = -1 + \frac{\exp(- \lambda_2 \sd_2(x, y_2))}{\beta_2^*(x | y_1)}
\end{align}
 The right derivative of the expression in the left side of \eqref{eq:dif1} with respect to $\epsilon$ evaluated at $\epsilon = 0$ is 
\two{displayed in \eqref{eq:der} below}{given by
 \begin{align}
\E{\frac{\beta_2^* (X|  y_1)^{\frac{1}{1+\nu_1}}\exp \left(-\frac{\lambda_1}{1+\nu_1} \sd_1(X, y_1)  \right)}{\beta_1^*(X)}  \left(  1 - \frac{\exp \left(- \lambda_2 \sd_2(X, y_2)   \right) }{\beta_2^* (X|  y_1)} \right)} \geq 0, \label{eq:der}
\end{align}
 } and is equivalent to  \eqref{eq:sigma12a}. Note that bringing the differentiation inside of the expectation is allowed by the dominated convergence theorem: the difference quotient of the integrand in \eqref{eq:dif1} is proportional to $\frac{((1 - \epsilon) a + \epsilon b)^{\nu_2} - a^{\nu_2} }{\epsilon}$, for $a \geq 0$, $b \geq 0$, which is bounded below by $0$ and above by a constant times $a^{\nu_2 - 1} b$ in the range $\epsilon \leq \delta < 1$, for some $\delta$.

We proceed to show \eqref{eq:g1max}, which will imply the sufficiency part. We apply \thmref{thm:csiszarg} twice to write
\begin{align}
\two{&~}{} L \left( P_{Y_1|X}, P_{Y_2 | XY_1}\right) \two{\notag \\}{}
\geq&~ (1+\nu_1)I(X; Y_1) + \lambda_1 \E{\sd_1(X, Y_1)} \two{\notag\\
&~}{} +   I(X; Y_2 | Y_1)  + \lambda_2 \E{\sd_2(X, Y_2)} \label{eq:optcond5-}\\
\geq&~ (1+\nu_1)I(X; Y_1) + \lambda_1 \E{\sd_1(X, Y_1)} + \E{\log \frac 1 {\beta_2(X |Y_1)}} 
 \label{eq:optcond4}\\
\geq&~ (1+\nu_1) \E{\log \frac 1 {\beta_1(X)}} \label{eq:optcond5}
 \end{align}
where \eqref{eq:optcond4} holds for all $\beta_2(x| y_1) \geq 0$ satisfying
\begin{align}
\E{\frac{\exp(- \lambda_2 \sd_2(X, y_2))}{\beta_2(X| y_1)} | Y_1 = y_1} \leq 1 ~\forall (y_1, y_2) \in \mathcal Y_1 \times \mathcal Y_2
\label{eq:s21a}
\end{align}
with equality if and only if 
$(P_{ Y_2 | X,Y_1 = y_1}, \beta_2) = ( P_{Y_2^* | , X, Y_1^* = y_1}, \beta_2^*)$ for $P_{Y_1}$-a.e. $y_1$. 

Likewise, \eqref{eq:optcond5} holds for all $\beta_1(x) \geq 0$ satisfying \eqref{eq:sigma1}, with equality if and only if
\begin{align}
\frac{d P_{ X | Y_1 = y_1}}{dP_{X}}(x) &=  \frac {\beta_2(x| y_1)^{\frac{1}{1+\nu_1}}\exp \left(-\frac{\lambda_1}{1+\nu_1} \sd_1(x, y_1) \right)} {\beta_1(x)}  \label{eq:xy1}\\
 \beta_1(x) 
  &= \E{\beta_2(x| Y_1)^{\frac{1}{1+\nu_1}}\exp \left(-\frac{\lambda_1}{1+\nu_1} \sd_1(x, Y_1)  \right)} \label{eq:optcond5+}.
\end{align}

Substituting \eqref{eq:xy1} into \eqref{eq:s21a}, we obtain \eqref{eq:sigma12}, and \eqref{eq:g1max} follows, together with condition for equality. 
\end{proof}

  \two{
  \vspace*{20pt}
 \begin{strip}
\normalsize
\hrulefill
\begin{align}
\E{\frac{\beta_2^* (X|  y_1)^{\frac{1}{1+\nu_1}}\exp \left(-\frac{\lambda_1}{1+\nu_1} \sd_1(X, y_1)  \right)}{\beta_1^*(X)}  \left(  1 - \frac{\exp \left(- \lambda_2 \sd_2(X, y_2)   \right) }{\beta_2^* (X|  y_1)} \right)} \geq 0, \label{eq:der}
\end{align}
\hrulefill
\end{strip}
}{}

\section{Iterative algorithm}
\label{sec:algo}
\subsection{Computation of single stage rate-distortion function}
In the context of finite source and reproduction alphabets, an algorithm for computation of rate-distortion functions was proposed by Blahut \cite{blahut1972computation}. Below, we state it for general alphabets in \algref{algo:blahut} and provide its convergence analysis in \thmref{thm:blahutconver}. In \secref{sec:blahutsucc} below, we generalize these results to successive refinement.  
\begin{algorithm}
\SetKwInOut{Input}{input}\SetKwInOut{Output}{output}
\Input{ $\lambda > 0$; maximum number of iterations $K$.}
\Output{ An estimate $F_K(\lambda)$ of $F(\lambda)$, the Lagrange dual of $R(d)$, defined in \eqref{eq:Fl}.
}
Fix $P_{Y_0}$.\\
\For{$ k = 1, 2, \ldots, K$}{

Compute, using \eqref{eq:sigmabarY}, $\Sigma_{Y_{k-1}}(x)$\;
Compute, using \eqref{eq:tilted}, the transition probability kernel $P_{Y_{k} | X}$ that achieves the minimum in 
\begin{equation}
F_k(\lambda) =  \min_{P_{Y|X}}   L(P_{Y|X}, P_{Y_{k-1}}),
\end{equation}
and $L$ is defined in \eqref{eq:Lbarsc}\; 
Record the corresponding output distribution $P_{Y_k}$, $P_X \to P_{Y_{k} | X} \to P_{Y_{k}}$\; 

}
\caption{The generalized Blahut algorithm.}\label{algo:blahut}
\end{algorithm}

\begin{thm}
Suppose $Y^*$ attains the minimum in \eqref{eq:Fl}, and let $Y_0$ be such that  $D(Y^* \| Y_0) < \infty$. Consider \algref{algo:blahut}. The sequence $F_k(\lambda)$ is monotonically decreasing to $F(\lambda)$, and the convergence speed is bounded as
 \begin{equation}
 F_k(\lambda) - F(\lambda) \leq \frac{D(Y^* \| Y_0)}{k}. \label{eq:convspeed}
\end{equation}
\label{thm:blahutconver}
\end{thm}

\begin{proof}
The analysis below is inspired by Csisz\'ar \cite{csiszar1974computation}. From \eqref{sc:-dtilted1a}, we have
\begin{align}
F_{k-1}(\lambda)
&= L(P_{Y_k|X}, P_{Y_{k}}) + D(Y_k \| Y_{k-1}) \label{eq:Lkk}\\
&\geq F_k(\lambda) + D(Y_k \| Y_{k-1}),
\end{align}
 and 
\begin{align}
 F_{k}(\lambda) \leq F_{k-1}(\lambda), \label{eq:decreasing}
\end{align}
 with equality if and only if $Y_k = Y_{k-1}$, which implies that $F_{k-1}(\lambda) = F_{k}(\lambda) = F(\lambda)$. 
 
 Taking an expectation of \eqref{eq:tilted} (particularized to $\bar Y = Y_{k-1}$) with respect to $P_{X Y^*}$, we conclude
\begin{align}
F_k(\lambda) &= F(\lambda) + D(Y^* \| Y_{k-1}) - D(P_{Y^*|X} \| P_{Y_{k} | X} | P_X) \\
&\leq F(\lambda) + D(Y^* \| Y_{k-1}) - D(Y^*\| Y_k)  \label{eq:dp}%
\end{align}
where \eqref{eq:dp} holds by the data processing inequality for relative entropy.  

 To show \eqref{eq:convspeed}, we apply \eqref{eq:decreasing} and \eqref{eq:dp}  as follows. 
\begin{align}
K F_K(\lambda) - K F(\lambda) &\leq  \sum_{k=1}^K F_{k}(\lambda) - K F(\lambda)\\
&\leq \sum_{k = 1}^K ( D(Y^* \| Y_{k-1}) - D(Y^* \| Y_{k}) )\\
&= D(Y^*\| Y_0) - D(Y^*\| Y_K). 
\end{align}
\end{proof}

Note that $D(Y^* \| Y_0) < \infty$ is a sufficient condition for convergence of \algref{algo:blahut}. This condition is trivially satisfied if the reproduction alphabet is finite and $P_{Y_0}$ is supported everywhere. 
 
An alternative convergence guarantee can be obtained as follows. Considering \eqref{eq:dp} and noting that
\begin{equation}
 D(Y^* \| Y_{k-1}) - D(Y^*\| Y_k) \leq \sup_{y \in \mathcal Y}\, \log \frac{dP_{Y_{k}}}{dP_{Y_{k-1}}}(y), \label{eq:stop}
\end{equation}
 we can employ the following stopping criterion for the Blahut algorithm to guarantee estimation accuracy $\delta$:
 if  $\sup_{y \in \mathcal Y}\, \log \frac{dP_{Y_{k}}}{dP_{Y_{k-1}}}(y) \leq \delta$, then stop and output $\tilde F(\lambda) = F_k(\lambda)$. If the same stopping rule is applied for all $\lambda \geq 0$, using \eqref{eq:lagr1}, we find that
the corresponding estimate of the rate-distortion function $\tilde R(d)$ satisfies the same accuracy guarantee:
\begin{equation}
R(d) \leq  \tilde R(d) \leq R(d) + \delta.
\end{equation}

 \subsection{Computation of the rate-distortion function for successive refinement}
 \label{sec:blahutsucc}

 A generalization of discrete Blahut's algorithm to successive refinement is proposed in \cite{tuncel2003computation}. \algref{algo:blahutsucc} presents a generalization of the algorithm to abstract alphabets, and \thmref{thm:blahutsuccconver} presents its convergence analysis.

 \begin{algorithm}
\SetKwInOut{Input}{input}\SetKwInOut{Output}{output}
\Input{ $(\nu_1, \lambda_1, \lambda_2) > 0$; maximum number of iterations $K$.}
\Output{ An estimate $F_K(\nu_1, \lambda_1, \lambda_2)$ of $F(\nu_1, \lambda_1, \lambda_2)$, the Lagrange dual of $R_2(d_1, d_2, R_1)$, defined in \eqref{eq:Lstar}.
}
Fix $P^0_{Y_{1}}$ and $P^0_{Y_2 | Y_1}$\;
\For{$ k = 1, 2, \ldots, K$}{
Compute 
\begin{align}
\label{eq:g2starc}
 \two{&}{}\beta_2^{k-1}(x| y_1) 
 \two{\\=\,& \notag}{&=} 
 \mathbb E_{P^{k-1}_{Y_2 | Y_1 = y_1}} \left[ \exp(- \lambda_2 \sd_2(x, Y_2)) | Y_1 = y_1 \right],   \\
  \label{eq:g1starc}
\two{&}{} \beta_1^{k-1}(x) 
 \two{\\=\,&\notag} {&=}
\mathbb E_{P_{Y_1}^{k-1}} \left[ \beta_2^{k-1}(x| Y_1 )^{\frac{1}{1+\nu_1}}\exp \left(-\frac{\lambda_1}{1+\nu_1} \sd_1(x, Y_1)  \right)\right];
\end{align}
Using
\begin{align}
 \frac{d P^k_{ Y_1 | X = x}(y_1)}{dP^{k-1}_{Y_1}(y_1)} \two{}{&}=  \frac{\exp \left(-\frac{\lambda_1}{1+\nu_1} \sd_1(x, y_1) \right)}{\beta_1^{k-1}(x)\beta_2^{k-1}(x| y_1)^{-\frac{1}{1+\nu_1}}} , \two{&}{}   \label{eq:Y1k}\\
 \frac{d P^k_{ Y_2 | X = x,Y_1 = y_1 }(y_2)}{dP^{k-1}_{Y_2 | Y_1 = y_1}(y_2)} \two{}{&}=  \frac {\exp (-\lambda_2 \sd_2(x, y_2))} {\beta_2^{k-1}(x| y_1)},   \two{&}{} \label{eq:Y2k}
\end{align}
compute the transition probability kernels $P^k_{Y_{1} | X}$ and $P_{Y_2 | X Y_1}^k$ that achieve the minimum in 
\begin{align}
\two{&~}{} F_k(\nu_1, \lambda_1, \lambda_2)  \two{\\ \notag}{}
&=  \min_{ P_{Y_1|X},\, P_{Y_2 | XY_1} }  \bar{L} \left( P_{Y_1|X}, P_{Y_2 | XY_1}, P_{Y_1}^{k-1}, P_{Y_2 | Y_1}^{k-1} \right) \\
&= (1+\nu_1)\E{\log \frac{1}{\beta_1^{k-1}(X)}};
\end{align}
where $\bar{L}$ is defined in \eqref{eq:Lbar}, and the minimum is computed in \eqref{eq:g1bar0}. Compute the corresponding $P_{Y_1}^k$ and $P_{Y_2 | Y_1}^k$\;

}
\caption{The generalized Blahut algorithm for successive refinement.}\label{algo:blahutsucc}
\end{algorithm}

\begin{thm}
Suppose $(P_{Y_1}^*, P_{Y_2^* | Y_1^*})$ attain the minimum in \eqref{eq:Lstar}, and let $P_{Y_1}^0$ and $P_{Y_2 | Y_1}^{0}$ be such that $D(P_{Y_1^*}\| P_{Y_1}^0)<\infty$ and $D(P_{Y_2^* | Y_1^*}\| P_{Y_2 | Y_1}^{0} |P_{Y_1^*}) < \infty$. Consider \algref{algo:blahutsucc}. The sequence $F_k(\nu_1, \lambda_1, \lambda_2)$ is monotonically decreasing to $F(\nu_1, \lambda_1, \lambda_2)$, and the convergence speed is bounded as
 \begin{align}
 \label{eq:convspeedsucc}
 \two{&~}{} F_k(\nu_1, \lambda_1, \lambda_2) - F(\nu_1, \lambda_1, \lambda_2) \two{\\}{}
 \leq&~ \frac 1 k \left( (1 + \nu_1) D(P_{Y_1^*}\| P_{Y_1}^0)  + D(P_{Y_2^* | Y_1^*}\| P_{Y_2 | Y_1}^{0} |P_{Y_1^*}) \right). \two{\notag}{}
\end{align}
\label{thm:blahutsuccconver}
\end{thm}

\begin{proof}
We build upon the ideas in the proof of \thmref{thm:blahutconver}. 
From the definition of $\bar L$ and $P_{Y_1|X}^k$, $P_{Y_2|XY_1}^k$, we have
\begin{align}
\label{eq:Lkksucc}
F_{k-1}
&=\bar L( P_{Y_1|X}^k, P_{Y_2 | XY_1}^k, P_{Y_1}^{k-1}, P_{Y_2 | Y_1}^{k-1}) 
\two{\\&\phantom{=}\notag}{}+ D(P_{Y_2 | Y_1}^k \| P_{Y_2 | Y_1}^{k-1} \| P_{Y_1}^k ) + (1 + \nu_1) D(P_{Y_1}^k \| P_{Y_1}^{k-1}) \\
&\geq F_k 
\two{\\&\phantom{=}\notag}{}
+ D(P_{Y_2 | Y_1}^k \| P_{Y_2 | Y_1}^{k-1} \| P_{Y_1}^k ) + (1 + \nu_1) D(P_{Y_1}^k \| P_{Y_1}^{k-1}),
\end{align}
where we suppressed the dependence of $F_k$ on $(\nu_1, \lambda_1, \lambda_2)$ for brevity, i.e. $F_k = F_k (\nu_1, \lambda_1, \lambda_2)$.  It follows that 
\begin{align}
 F_{k} \leq F_{k-1}, \label{eq:decreasingsucc}
\end{align}
 with equality if and only if $P_{Y_2 | Y_1}^{k-1} P_{Y_1}^{k-1} = P_{Y_2 | Y_1}^{k} P_{Y_1}^k$, which implies that $F_{k-1} = F_{k} = F$.   
 
 Taking expectations of the logarithms of \eqref{eq:Y1k} and \eqref{eq:Y2k}  with respect to $P_{X Y_1^* Y_2^*}$ and using \eqref{eq:g1bar0}, we deduce that
 \begin{align}
\two{&~}{} F_k \two{\notag\\}{}
=&~ F + (1+\nu_1) \left( D(P_{Y_1^*} \| P_{Y_1}^{k-1}) - D(P_{Y_1^*|X} \| P_{Y_1^{k} | X} \| P_X) \right) \notag\\
&~+ D(P_{Y_2^* | Y_1^*}\| P_{Y_2 | Y_1}^{k-1} |P_{Y_1^*})
\two{\notag\\&~}{}
 -D(P_{Y_2^* | X_1 Y_1^*}\| P_{Y_2 | X_1 Y_1}^{k} | P_{X_1 Y_1^*})  \\
\leq&~ F + (1+\nu_1) \left( D(P_{Y_1^*} \| P_{Y_1}^{k-1}) - D(P_{Y_1^*} \| P_{Y_1}^{k} ) \right)  \label{eq:dpsucc}\\
&~+ D(P_{Y_2^* | Y_1^*}\| P_{Y_2 | Y_1}^{k-1} |P_{Y_1^*}) -D(P_{Y_2^* | Y_1^*}\| P_{Y_2 | Y_1}^{k} | P_{Y_1^*}) \notag
\end{align}
where \eqref{eq:dpsucc} holds by the data processing inequality for relative entropy.  

 To show \eqref{eq:convspeedsucc}, we apply \eqref{eq:decreasingsucc} and \eqref{eq:dpsucc}  as follows. 
\begin{align}
\two{&~}{}
K F_K - K F 
\two{\notag\\}{}
\leq&~  \sum_{k=1}^K F_{k} - K F\\
\leq&~  \sum_{k = 1}^K \Big[ (1 + \nu_1) \left(D(P_{Y_1^*} \| P_{Y_1}^{k-1}) - D(P_{Y_1^*|X} \| P_{Y_1}^{k} )\right) \\
&~+  D(P_{Y_2^* | Y_1^*}\| P_{Y_2 | Y_1}^{k-1} |P_{Y_1^*}) -D(P_{Y_2^* | Y_1^*}\| P_{Y_2 | Y_1}^{k} | P_{Y_1^*})  \Big] \notag\\
=&~ (1 + \nu_1)(D(P_{Y_1^*}\| P_{Y_1}^0) - D(P_{Y_1^*}\| P_{Y_1}^K)) \\
&~+ D(P_{Y_2^* | Y_1^*}\| P_{Y_2 | Y_1}^{0} |P_{Y_1^*}) - D(P_{Y_2^* | Y_1^*}\| P_{Y_2 | Y_1}^{K} | P_{Y_1^*}) \notag
\end{align}
\end{proof}

Using \eqref{eq:dpsucc}, we can obtain the following analog of the stopping criterion in \eqref{eq:stop}: to achieve accuracy $F_k(\nu_1, \lambda_1, \lambda_2) - F(\nu_1, \lambda_1, \lambda_2) \leq \delta$, stop as soon as 
\begin{equation}
 \sup_{(y_1, y_2) \in \mathcal Y_1 \times \mathcal Y_2}\, (1 + \nu_1)\log \frac{dP_{Y_{1}}^k}{dP_{Y_1}^{k-1}}(y) + \frac{dP_{Y_2 | Y_1}^k}{dP_{Y_2 | Y_1}^{k-1}}(y) \leq \delta. \label{eq:stopsuc}
\end{equation}
For finite alphabet sources, a counterpart of \eqref{eq:stopsuc} was proposed by Tuncel and Rose \cite{tuncel2003computation}.

\subsection{Numerical example}
Consider successive refinement of $X \sim \mathcal N(0,1)$ under squared error distortion. As is well known, Gaussian source under squared distortion is successively refinable \cite{equitz1991successive}, so at any $0 < d_2 \leq d_1 \leq 1$ and $R(d_1) \leq R_1$, $R_2(d_1, d_2, R_1) = R(d_2) = \frac 1 2 \log \frac 1 {d_2}$.  

In this experiment, we ran \algref{algo:blahutsucc} to verify that it computes an estimate of $R_2(d_1, d_2, R_1)$ that closely matches $R(d_2)$. 

We fixed $\lambda_1 = 5/9$, which corresponds to $d_1 = 0.9$. We also fixed $\nu_1 = 1$ (for this example, the choice of $\nu_1> 0$ is immaterial and can be chosen arbitrarily, as per discussion after \eqref{eq:l2sr}). We set starting densities $P_{Y_1}^0$ and $P_{Y_2|Y_1 = y_1}^0$ to be $\mathcal N(0, 1)$ and $\mathcal N(y_1, 1)$, respectively, ensuring that all the densities in \algref{algo:blahutsucc} are Gaussian, and all the integrals can be computed in closed form. We chose $31$ exponentially spaced slope samples $\lambda_2 > 0$,  and we ran the algorithm for the maximum of $K = 20$ iterations at each choice of $\lambda_2$. In \figref{fig:reg}, 31 straight lines of slopes $-\lambda_2$ correspond to $F_K - \lambda_2 d_2 - \lambda_1 d_1 - \nu_1 R_1$. Their upper convex envelope is the numerical estimate of $R_2(d_1, d_2, R_1)$ according to the algorithm. In \figref{fig:regg}, it is undistinguishable from the the thick curve, which represents the theoretical minimum total rate, $\frac 1 2 \log \frac 1 {d_2}$.

 Computing the expectations in Algorithms \ref{algo:blahut} and \ref{algo:blahutsucc} is easy to do if the output alphabets are finite, even if $P_X$ is continuous, a case also not previously addressed in literature. For infinite output alphabets, computing these expectations can be a computational bottleneck. Still, one could use Algorithms \ref{algo:blahut} and \ref{algo:blahutsucc} to look for the best approximation within a certain family of distributions parametrized by a finite number of parameters. The quality of the approximation will depend on how appropriately the parametric family is chosen. To choose a good family, one could look for a separate theoretical argument that would ensure that the infimum is attained within some class of distributions. Theorems \ref{thm:blahutconver} and \ref{thm:blahutsuccconver} would then ensure convergence when running the algorithm within that class.

\begin{apxonly}
Computations in the Gaussian case. 
\begin{align}
M(a, t) &\triangleq e^{-at/(1+2t)} (1 + 2 t)^{- \frac 1 2};\\
t_2 &=  \sigma_{2|1}^2 \lambda_2\\
\beta_2(x | y_1) &= M((x- y_1)^2/\sigma_{2|1}^2, t_2)\\
&= e^{-(x- y_1)^2 \frac {t_2}{\sigma_{2|1}^2 (1+2t_2)}} (1 + 2 t_2)^{- \frac 1 2}\\
t_1 &= \frac {\sigma_1^2} {1 + \nu_1} \left( \frac{t_2 }{\sigma_{2|1}^2(1+2t_2)} + \lambda_1 \right)\\
\beta_1(x) &=  (1 + 2 t_2)^{- \frac 1 {2(1 + \nu_1)}} M(x^2 / \sigma_1^2, t_1) \\
&= (1 + 2 t_2)^{- \frac 1 {2(1 + \nu_1)}} (1 + 2t_1)^{-\frac 1 2} e^{-\frac {x^2 t_1}{\sigma_1^2(1 + 2t_1)}}\\
Y_1 | X &= x \sim \mathcal N\left( \frac{2t_1}{1 + 2t_1} x, \frac{\sigma_1^2}{1 + 2 t_1}\right)\\
\sigma_{1, new}^2 &= \left(\frac{2t_1}{1 + 2t_1}\right)^2 + \frac{\sigma_1^2}{1 + 2 t_1} \\
X |Y_1  &\sim \mathcal N\left( \frac{2t_1}{(1 + 2t_1)\sigma_{1, new}^2} y_1, 1 - \left(\frac{2t_1}{(1 + 2 t_1)\sigma_{1, new}}\right)^2\right)\\
Y_2 | Y_1, X &\sim \mathcal N\left( \frac{2t_2 x + y_1}{1 + 2t_2}, \frac{\sigma_{2|1}^2}{1 + 2 t_2}\right)\\
Y_2 | Y_1 &\sim \mathcal N\left( \frac{y_1}{1 + 2t_2} + \frac{2t_2}{1 + 2t_2} \E{X|Y_1 = y_1}, \left(\frac{2t_2}{1 + 2t_2} \right)^2 \sigma_{X|Y_1}^2 + \frac{\sigma_{2|1}^2}{1 + 2 t_2} \right)
\end{align}
\end{apxonly}

\begin{figure}[htp]
\centering
\epsfig{file=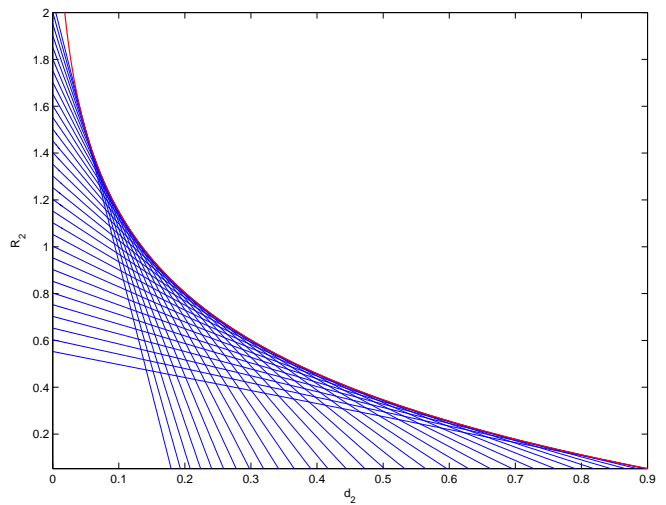,width=\two{1}{.5}\linewidth}
 \caption[]{The minimum total rate at stage $2$ for Gaussian successive refinement, for fixed $\lambda_1 = 5/9$ (corresponding to $d_1 = 0.9$, $R_1 = - .5 \log .9 \approx .05 $). } \label{fig:regg}
\end{figure}

\section{Conclusion}
In this paper, we revisited the parametric representation of rate-distortion function of abstract sources (\thmref{thm:csiszarg}, proof in \secref{sec:csiszar}). We showed its generalization to the successive refinement problem (\thmref{thm:parametric}, proof in \secref{sec:proof}). That representation leads to a tight nonasymptotic converse bound for successive refinement, presented in \secref{sec:nonasymptotic}. It also helps to formulate and prove the convergence of an iterative algorithm that can be applied to compute the rate-distortion function on abstract alphabets, presented in \secref{sec:algo}. 

It will be interesting to see whether the approach presented in this paper can be applied to study rate-distortion regions of other important multiterminal information theory problems, such as lossy compression with side information available at decoder (the Wyner-Ziv problem \cite{wynerziv1976rate}), the multiple descriptions problem \cite{wolf1980source} and lossy compression with possibly absent side information (the Kaspi problem \cite{kaspi1994rate}).
It also paves the way to  a refined nonasymptotic analysis of successive refinement for abstract sources.

\section{Acknowledgement}
We would like to thank Lin Zhou for valuable comments regarding a second-order analysis of the nonasymptotic bounds in \secref{sec:nonasymptotic}, and both anonymous reviewers for detailed suggestions. 

\bibliographystyle{IEEEtran}
\bibliography{../../Bibliography/rateDistortion,../../Bibliography/vk}

\end{document}